\documentclass[11pt]{article}
\usepackage{fullpage}
\usepackage[utf8]{inputenc}
\usepackage{braket}
\usepackage{amsthm}
\usepackage[pagebackref]{hyperref}
\usepackage{amsmath, amssymb}
\usepackage[pdftex]{graphicx}
\usepackage{authblk}
\usepackage{xcolor}
\usepackage{algorithm}
\usepackage{mathtools}
\usepackage{algpseudocode}
\usepackage{latexsym}

\hypersetup{
    linktocpage,
    colorlinks=true,
    citecolor=blue,
    linkcolor=red,
    urlcolor=orange,
}

\newcommand{\BQP}{\mathsf{BQP}}
\newcommand{\NP}{\mathsf{NP}}
\newcommand{\MA}{\mathsf{MA}}
\newcommand{\QMA}{\mathsf{QMA}}
\newcommand{\dQMA}{\mathsf{dQMA}}
\newcommand{\dMA}{\mathsf{dMA}}
\newcommand{\QMAcc}{\mathsf{QMAcc}}
\newcommand{\EQ}{\mathsf{EQ}}
\newcommand{\poly}{\mathrm{poly}}
\newcommand{\GT}{\mathsf{GT}}

\newcommand{\dQMAsep}{\mathsf{dQMA}^\mathsf{sep}}
\newcommand{\dQMAsepsep}{\mathsf{dQMA}^{\mathsf{sep},\mathsf{sep}}}

\title{On the Power of Quantum Distributed Proofs}

\author[1]{Atsuya Hasegawa\thanks{atsuyahasegawa@is.s.u-tokyo.ac.jp}}
\author[2]{Srijita Kundu\thanks{srijita.kundu@uwaterloo.ca}}
\author[3]{Harumichi Nishimura\thanks{hnishimura@i.nagoya-u.ac.jp}}
\affil[1]{\textit{Graduate School of Information Science and Technology, The University of Tokyo, Japan}}
\affil[2]{\textit{Institute for Quantum Computing and Department of Combinatorics and Optimization, University of Waterloo, Canada}}
\affil[3]{\textit{Graduate School of Informatics, Nagoya University, Japan}}

\date{}

\newtheorem{theorem}{Theorem}
\newtheorem{definition}{Definition}
\newtheorem{lemma}[theorem]{Lemma}
\newtheorem{claim}[theorem]{Claim}
\newtheorem{fact}{Fact}
\newtheorem{corollary}[theorem]{Corollary}

\newtheorem{proposition}[theorem]{Proposition}

\begin{document}

\maketitle

\begin{abstract}
Quantum nondeterministic distributed computing was recently introduced as $\mathsf{dQMA}$ (distributed quantum Merlin-Arthur)  protocols by Fraigniaud, Le Gall, Nishimura and Paz (ITCS 2021). In $\mathsf{dQMA}$ protocols, with the help of quantum proofs and local communication, nodes on a network verify a global property of the network. Fraigniaud et al. showed that, when the network size is small, there exists an exponential separation in proof size between distributed classical and quantum verification protocols, for the equality problem, where the verifiers check if all the data owned by a subset of them are identical. In this paper, we further investigate and characterize the power of the $\mathsf{dQMA}$ protocols for various decision problems.

First, we give a more efficient $\mathsf{dQMA}$ protocol for the equality problem with a simpler analysis. This is done by adding a symmetrization step on each node and exploiting properties of the permutation test, which is a generalization of the SWAP test. We also show a quantum advantage for the equality problem on path networks still persists even when the network size is large, by considering ``relay points'' between extreme nodes. 

Second, we show that even in a general network, there exist efficient $\mathsf{dQMA}$ protocols for the ranking verification problem, the Hamming distance problem, and more problems that derive from efficient quantum one-way communication protocols. Third, in a line network, we construct an efficient $\mathsf{dQMA}$ protocol for a problem that has an efficient two-party $\mathsf{QMA}$ communication protocol.

Finally, we obtain the first lower bounds on the proof and communication cost of $\mathsf{dQMA}$ protocols. To prove a lower bound on the equality problem, we show any $\mathsf{dQMA}$ protocol with an entangled proof between nodes can be simulated with a $\mathsf{dQMA}$ protocol with a separable proof between nodes by using a $\mathsf{QMA}$ communication-complete problem introduced by Raz and Shpilka (CCC 2004).
\end{abstract}

\clearpage
\tableofcontents
\clearpage
\section{Introduction}

\subsection{Background}

\subsubsection*{Quantum distributed computing}
Quantum distributed computing is the quantum analog of distributed computing where parties are quantum computers and communication in a network is done via qubits. A few early works initiated the study of quantum distributed computing \cite{BOH05,TKM12,GKM09,EKNP14}. See also \cite{BT08,DP08,AF14} for general discussions. 

Recently, the quantum distributed computing model has been intensively studied to identify quantum advantages in the number of rounds and the amount of communication in distributed computing. The major models in classical distributed computing have been explored since the seminal work by Le Gall and Magniez \cite{GM18}; $\mathsf{CONGEST}$ model \cite{GM18,IGM20,MN22,CFGLO22,vAdV22,WY22}, $\mathsf{CONGEST}$-$\mathsf{CLIQUE}$ model \cite{IG19} and $\mathsf{LOCAL}$ model \cite{GNR19,GR22,CRdG+23}.

\subsubsection*{Nondeterministic distributed computing}

For both theoretical and application reasons, on distributed networks, it is quite important to efficiently verify some global properties of the network with local (i.e., constant-round) communication. The most widely accepted and studied criteria for distributed verification is as follows \cite{Fra10}:
\begin{itemize}
    \item (completeness) For a $\textit{yes}$-instance, all the nodes must accept.
    \item (soundness) For a $\textit{no}$-instance, at least one node must reject.
\end{itemize}
Intuitively, if the global property of the graph is appropriate, all the nodes are satisfied, and otherwise, at least one node raises an alarm to all the other nodes.

On the other hand, many properties cannot be checked with such local communication, and usually require many rounds on the networks. A possible extension is to give information to the nodes on the network. Such a scheme was introduced as proof-labelling schemes \cite{KKP10} and locally checkable proofs \cite{GS16}, which are considered distributed $\NP$ protocols. More recently, randomized proof-labeling schemes were introduced \cite{FPSP19}, and these protocols are considered as \emph{distributed Merlin-Arthur} ($\dMA$) protocols. In a $\dMA$ protocol, an untrusted prover sends a classical proof to all the nodes on a network. Based on their part of the proof, each node, who can use a randomized algorithm, on the network simultaneously sends messages to its neighbors and receives messages from its neighbors in constantly many rounds. Finally each node outputs accept or reject in a probabilistic manner so that completeness is high, i.e., the completeness condition holds with probability at least, say, $\frac{2}{3}$ (completeness $\frac{2}{3}$) and soundness error is low, i.e., the soundness condition does not hold with probability at most, say, $\frac{1}{3}$ (soundness $\frac{1}{3}$).

While a $\dMA$ protocol is more powerful than usual deterministic distributed computing, unfortunately, there are still limits on this model for some predicates \cite{FGNP21}.

\subsubsection*{Distributed quantum Merlin-Arthur ($\dQMA$) protocols}

Fraigniaud, Le Gall, Nishimura, and Paz \cite{FGNP21} introduced the setting where a prover and nodes are \emph{quantum} computers and communicate with \emph{quantum} messages, and named such the protocols distributed \emph{quantum} Merlin-Arthur ($\dQMA$) protocols.

The global property they considered was the problem $\EQ$ of deciding whether all the distributed data ($n$-bit binary strings) on the network are the same or not. The basic idea behind their $\dQMA$ protocol is to make the prover send quantum fingerprints for the input data \cite{BCWdW01} to all the nodes; subsequently, each node sends the fingerprint it receives to its neighbor, and they do the SWAP test \cite{BCWdW01}, a quantum procedure for checking whether two quantum fingerprints are the same or not. 
Their $\dQMA$ protocol needs local proof size $O(tr^2\log n)$, namely, each node receives an $O(tr^2\log n)$-qubit proof, where $r$ is the radius of the network and $t$ is the number of distributed inputs. 
As a complementary result to their $\dQMA$ upper bound, they showed that any $\dMA$ protocol with high completeness and low soundness error requires an $\Omega(n)$ size classical proof for at least one node. As a consequence, they gave an exponential gap in the proof size between $\dMA$ protocols and $\dQMA$ protocols for the equality problem.

They also derive an efficient $\dQMA$ protocol on a path, for any function that has an efficient quantum one-way communication protocol with bounded error in the communication complexity setting. As a corollary, they have an efficient $\dQMA$ protocol on a path for the Hamming distance problem since it has an efficient quantum one-way communication protocol \cite{Yao03}.

The results of \cite{FGNP21} are summarized in Table \ref{tab:fgnp21}, where $\#$Terminals represents the number of terminals, the nodes that have distributed inputs. For a function $f:(\{0,1\}^n)^2 \to \{0,1\}$, let us denote by $\BQP^1(f)$ the quantum one-way communication complexity of $f$.

\begin{table}[hbtp]
    \centering
    \begin{tabular}{cccccc}
         Protocol & Problem & $\#$Terminals & Round Number & Local Proof Size & \\
         \hline \hline
         Quantum & $\EQ$ & $t$ & 1 & $O(tr^2 \log n)$ & \\
         Quantum & $f$ & 2 & 1 & $O(r^2 \BQP^1(f) \log (n+r))$ & \\
         Classical & $\EQ$ & 2 & $\nu$ & $\Omega(\frac{n}{\nu})$ & \\
    \end{tabular}
    \caption{Summary of the results by Fraigniaud, Le Gall, Nishimura, and Paz \cite{FGNP21}}
    \label{tab:fgnp21}
\end{table}

\subsection{Our results}

In this work, we further investigate the power and limits of $\dQMA$ protocols, and give a comprehensive characterization for various decision problems.

\subsubsection*{Improved $\dQMA$ protocols for $\EQ$}

We derive a more efficient $\dQMA$ protocol for $\EQ$ on a general graph with multiple input terminals, by a simpler analysis of soundness. Our protocol and analysis are simpler than the ones in \cite{FGNP21} and the proof size of our $\dQMA$ protocol does not depend on the number of the terminals and matches the size of the path case with two terminals. 

\begin{theorem}[Theorem \ref{theorem:eq_tree}]
    There exists a $\dQMA$ protocol for $\mathsf{EQ}$ between $t$ terminals, on a network of radius $r$, with perfect completeness (i.e., completeness $1$) and sufficiently low soundness error, using local proof and message of size $O(r^2 \log n)$.
\end{theorem}

The result of \cite{FGNP21} implies that there is an exponential difference in proof size between $\dMA$ and $\dQMA$ for $\EQ$ on a path. However, such a big difference holds only when the network size is much smaller than the input size, i.e., $r \ll n$. Since there exists a trivial classical protocol with $n$-bit proofs (the prover sends the whole $n$-bit string to all the nodes, and each node checks if the proofs of its neighbors are identical to its own or not), the quantum strategy can be even worse than the trivial classical strategy when the network size is not so small.

In this paper, we show that even when the network size is not so small, a provable quantum advantage still persists. To claim the quantum advantage, we consider the complexity measure of the total size of proofs to all the nodes rather than the size of respective proofs to each node.

\begin{theorem}[Informal version of Theorem \ref{theorem:robust_QA} and Corollary \ref{corollary:classical_lowerbound_for_EQ}]
There exists a $\dQMA$ protocol for $\EQ$ on the path of length $r$, with 1-round communication, perfect completeness and sufficiently low soundness error, and with $\Tilde{O}(rn^\frac{2}{3})$ qubits as proofs in total. In contrast, any $\dMA$ protocol for $\EQ$ with constant-round communication, sufficiently high completeness and low soundness error, requires $\Omega(rn)$ bits as proofs in total.
\end{theorem}

\subsubsection*{The power of $\dQMA$ protocols for various problems}

Checking how large an input is among all inputs held by the terminals in a network is a fundamental problem. We name this problem the ranking verification ($\mathsf{RV}$) problem, and show that there exists an efficient $\dQMA$ protocol for it.

\begin{definition}[Ranking verification problem, informal version of Definition \ref{definition:ranking}]
    For $i,j \in [1,t]$, $\mathsf{RV}^{i,j}_{t}(x_1,\ldots,x_t)=1$ if and only if $x_i$, which is held by the $i$-th terminal, is the $j$-th largest input among $t$ $n$-bit integers $x_1,\ldots,x_t$.
\end{definition}
\begin{theorem}[Informal version of Theorem \ref{theorem:ranking_verification}]
    There exists a $\dQMA$ protocol for $\mathsf{RV}$ between $t$ terminals on a network of radius $r$, with perfect completeness and sufficiently low soundness error, using local proof and messages of size $O(tr^2 \log n)$.
\end{theorem}

To prove this statement, we derive an efficient $\dQMA$ protocol on a path to solve the greater-than function. The greater-than function ($\mathsf{GT}$) is defined as $\GT(x,y)=1$ if and only if $x > y$, where $x$ and $y$ are $n$-bit integers.

\begin{theorem}[Theorem \ref{theorem:greater-than}]
    There exists a $\dQMA$ protocol for $\GT$ on the path of length $r$ with $1$-round communication, perfect completeness, and sufficiently low soundness error, using local proof and message of size $O(r^2 \log n)$.
\end{theorem}

We can show that any $\dMA$ protocol for $\GT$ with high completeness and low soundness error requires $\Omega(nr)$ size classical proofs in total. Thus, this provides us another fundamental problem that exhibits an exponential quantum advantage in distributed verification. 
 
The result of \cite{FGNP21} on converting a quantum one-way communication protocol to a $\dQMA$ protocol only works on a path with two inputs, and no efficient $\dQMA$ protocol was known for three or more inputs over general networks. We construct an efficient $\dQMA$ protocol on a general graph with multiple terminals, for any function which has an efficient quantum one-way communication protocol with bounded error. For a function $f:(\{0,1\}^n)^2 \to \{0,1\}$, we define the multi-input function $\forall_t f:(\{0,1\}^n)^t \to \{0,1\}$ where $\forall_t f(x_1,\ldots,x_t) =1$ iff $f(x_i,x_j)=1$ for any $i,j \in [1,t]$.

\begin{theorem}[Theorem~\ref{proposition:multi-input}]
    For a function $f:(\{0,1\}^n)^2 \to \{0,1\}$, there exists a 1-round $\dQMA$ protocol for $\forall_t f$ on a network of radius $r$, with sufficiently high completeness and low soundness error, using local proof and message of size $O(t^2 r^2 \mathsf{BQP}^{1}(f) \log(n+t+r))$.
\end{theorem}

We also construct an efficient $\dQMA$ protocol for a function which has an efficient $\QMA$ communication protocol (introduced by Raz and Shpilka~\cite{RS04}) rather than an efficient quantum one-way communication protocol. Let us denote by $\QMAcc(f)$ the sum of the proof and communication amount of $\QMA$ communication protocols for $f$. 

\begin{theorem}[Informal version of Proposition \ref{proposition:dQMA_for_QMAcc}]
    There exists a $\dQMA$ protocol to solve $f$ on the path of length $r$ with sufficiently high completeness and low soundness error, using local proof and message of size $O(r^2 \log (r) \poly (\QMAcc(f)))$.
\end{theorem}

In addition, we show that any $\dQMA$ protocol in which entangled proofs are given to the nodes can be simulated with a $\dQMA$ protocol with ``separable'' proofs, in which the local part of the proof at each node is not entangled with the other nodes, with some overheads. A $\dQMAsep$ protocol is a $\dQMA$ protocol whose completeness holds with a proof that is separable between nodes. 

\begin{theorem}[Informal version of Theorem \ref{theorem:dQMAsep_for_dQMA}]
    For a function $f$ which has a  constant-round efficient $\dQMA$ protocol on a path (with entangle proofs), there exists a $1$-round efficient $\dQMAsep$ protocol for $f$.
\end{theorem}

Our results on quantum upper bounds and classical lower bounds are summarized in Table \ref{tab:result_upper}. As seen in the table, all the $\dQMA$ protocols constructed in this paper are actually $\dQMAsep$ ones. Let $\dQMA(f)$ denote the sum of the total proof size and the communication size of a $\dQMA$ protocol for $f$.
\begin{table}[hbtp]
    \centering
    \begin{tabular}{ccccccc}
         Protocol & Problem & $\#$Terminals & Local Proof Size & Total Proof Size & Ref &\\
         \hline \hline
         $\dQMAsep$ & $\EQ$ & $t$ & $O(r^2 \log n)$ & & \S\,\ref{section:perm_eq} &\\
         $\dQMAsep$ & $\EQ$ & $2$ & $n$ or $O(r^2 \log n)$ & $\Tilde{O}(r n^{\frac{2}{3}})$ & \S\,\ref{subsection:robust_eq_quantum_upper_bound} &\\
         $\dMA$ & $\EQ, \GT$ & $2$ & & $\Omega(rn)$ & \S\,\ref{subsection:classical_lower_bound}&\\
         $\dQMAsep$ & $\GT$ & $2$ & $O(r^2 \log n)$ & & \S\,\ref{subsection:gt} &\\
         $\dQMAsep$ & $\mathsf{RV}$ & $t$ & $O(tr^2 \log n)$ & & \S\,\ref{subsection:rv} & \\
         $\dQMAsep$ & $\forall_t f$ & $t$ & $O(t^2 r^2 \BQP^1(f)  \log (n+t+r))$ & & \S\,\ref{section:hamming_distance} &\\
         $\dQMAsep$ & $f$ & $2$ & $O(r^2 \log (r) \poly (\QMAcc(f)) )$ & & \S\,\ref{section:dQMAsep}& \\
         $\dQMAsep$ & $f$ & $2$ & $\Tilde{O}(r^2 (\dQMA(f))^2 )$ & & \S\,\ref{section:dQMAsep}& \\
    \end{tabular}
    \caption{Summary of our results on quantum upper bounds and classical lower bounds}
    \label{tab:result_upper}
\end{table}

\subsubsection*{Lower bounds for $\dQMA$ protocols}

In this paper, we derive the first lower bounds on the proof and communication cost of $\dQMA$ protocols. We introduce a $\dQMAsepsep$ protocol as another variant of $\dQMA$ protocols where a prover can only send separable proofs between nodes (and thus soundness holds only with respect to separable proofs). When we restrict the power of the prover, we obtain the following strong lower bound (note that it implies the matching lower bounds for $\EQ$ and $\GT$ with respect to the order of the input size $n$ as their sizes of $1$-fooling sets are $2^n$).

\begin{theorem}[Informal version of Theorem \ref{theorem:quantum_lower_bound_separable}]
    Let $\nu \in \mathbb{N}$ be a constant and $f:(\{0,1\}^n)^2\rightarrow\{0,1\}$ be a Boolean function with a $1$-fooling set of size $2^n$ (the definition of $1$-fooling sets is given in Section~\ref{subsection:preliminary_communication_complexity}). Let $\mathcal{P}$ be a $\dQMAsepsep$ protocol for $f$ on the path of length $r$ with $\nu$-round communication, sufficiently high completeness and low soundness error. 
    Then, the total proof size is $\Omega(r\log n)$.
\end{theorem}

It is notoriously hard to prove lower bounds when dealing with entanglement between parties, and the seminal example is the case of $\mathsf{MIP^*}$ \cite{CHTW04,IV12,NW19,JNVWY21}. In $\dQMA$ protocols, nodes on a network might exploit the power of entangled proofs from a prover by clever local communication and computations. Despite this difficulty, we prove several lower bounds of $\dQMA$ protocols. The main result is as follows.

\begin{theorem}[Informal version of Theorem \ref{theorem:quantum_lower_bound_entangled2}]
    Let $f:(\{0,1\}^n)^2\rightarrow\{0,1\}$ be a Boolean function with a $1$-fooling set of size $2^n$ (including $\EQ$ and $\GT$). Let $\mathcal{P}$ be a $\dQMA$ protocol for $f$ on the path of length $r$ with constant-round communication, sufficiently high completeness and low soundness error.
    Then, the total proof and communication size of $\mathcal{P}$ is $\Omega((\log n)^{1/4-\epsilon})$ for a sufficiently small constant $\epsilon>0$.\label{thmintro:quantum_lower_bound_entangled2}
\end{theorem}

Additionally, we prove a $\dQMA$ lower bound for functions which are hard for $\QMA$ communication protocols, in terms of the one-sided smooth discrepancy \cite{Kla11}. Let us denote by $\mathsf{sdisc}^1 (f)$ the one-sided smooth discrepancy of a function $f$; it was shown in \cite{Kla11} that $\mathsf{sdisc}^1$ is a lower bound on $\QMA$ communication complexity.

\begin{theorem}[Informal version of Theorem \ref{theorem:quantum_lower_bound_reduction}]
    Assume that  $\mathcal{P}$ is a $\dQMA$ protocol on a line of length $r$ with arbitrary rounds to solve $f$ with sufficiently high completeness and low soundness error. Then, the total proof and communication size of $
    \mathcal{P}$ is $\Omega(\sqrt{\log \mathsf{sdisc}^1 (f)})$.\label{thmintro:quantum_lower_bound_reduction}
\end{theorem}
Note that the above theorem does not give a nontrivial lower bound for the equality function, since this function has a constant-cost classical randomized communication protocol, and therefore $\mathsf{sdisc}^1(\EQ)$ is at most constant. Theorem~\ref{thmintro:quantum_lower_bound_entangled2} thus outperforms  Theorem~\ref{thmintro:quantum_lower_bound_reduction} for the $\EQ$ function.

Our results on lower bounds (including other ones than the above three theorems) are summarized in the following Table \ref{tab:result_lower}. In the table, $\epsilon>0$ is any small constant and $f^+$ is any non-constant Boolean function $f$. As functions which are hard for $\QMA$ communication protocols \cite{Kla11}, let us denote by $\mathsf{DISJ}$ the disjointness function, by $\mathsf{IP}$ the inner product function, by $P_\mathsf{AND}$ the pattern matrix \cite{She11} of the AND function. 
These lower bounds will be formally stated and proved in Section~\ref{section:lower_bound}. 

\begin{table}[hbtp]
    \centering
    \begin{tabular}{ccccc}
         Protocol & Problem & Round Number & Lower Bound &\\
         \hline \hline
         $\dQMAsepsep$ & $\EQ,\GT$ & constant & total proof size $\Omega(r \log n)$ & \\
         $\dQMA$ & $\EQ,\GT$ & constant & total proof \& communication size $\Omega(\frac{(\log n)^{\frac{1}{2}-\epsilon}}{r^{1+\epsilon}})$&\\
         $\dQMA$ & $f^+$ & constant & total proof size $\Omega(r)$ &\\
         $\dQMA$ & $\EQ,\GT$ & constant & total proof \& communication size $\Omega((\log n)^{\frac{1}{4}-\epsilon})$ & \\
         $\dQMA$ & $\mathsf{DISJ}$ & arbitrary & total proof \& communication size $\Omega(n^\frac{1}{3})$ &\\
         $\dQMA$ & $\mathsf{IP}$ & arbitrary & total proof \& communication size  $\Omega(n^\frac{1}{2})$ &\\
         $\dQMA$ & $P_\mathsf{AND}$ & arbitrary & total proof \& communication size $\Omega(n^\frac{1}{3})$ & \\
    \end{tabular}   
    \caption{Summary of our results on quantum lower bounds}
    \label{tab:result_lower}
\end{table}

\subsection{Overview of our techniques}

\subsubsection*{Improved protocol for $\EQ$ with a simpler analysis and the permutation test}

In \cite{FGNP21}, they designed a protocol on a path where each node sends the received proof (quantum fingerprint) to its left neighbor with probability $\frac{1}{2}$, and thus the conditional probability that the SWAP test occurs is needed to analyze. To simplify the analysis of the soundness of the protocol, we add an extra step called the symmetrization step for each node. With this step, we can avoid using conditional probability because each node conducts the SWAP test with certainty. 

In the \cite{FGNP21} protocol for $\EQ$ with three or more terminals, every non-terminal node performs the SWAP test on the state that consists of the state received from the prover and a state randomly chosen from states received from the children. Every node discards the other states received from the children and are not used for the SWAP test. To improve the proof size of the protocol for general graphs from $O(t r^2 \log n)$ to $O(r^2 \log n)$, we make each node perform the permutation test \cite{BBD+97,BCWdW01,KNY08} on all the states from its children. 

The permutation test is a generalization of the SWAP test from 2-partite systems to $k$-partite systems for any integer $k \geq 2$. We identify the permutation test with a projector to the symmetric subspace of multiple systems as a special case of weak Schur Sampling~\cite{BCH06}. Using properties of Schur sampling, we show that, by using the permutation test, we can test how close the subspace is to given states.

\subsubsection*{Robust quantum advantage for $\EQ$ on a path}

To prove a universal quantum advantage for $\EQ$, we consider inserting multiple ``relay points'' per $O(n^\frac{1}{3})$ nodes between extreme nodes that receive $n$-qubit proofs. Based on the $n$-bit measurement results, nodes between relay points conduct the SWAP test-based quantum strategy. This makes for a $\dQMA$ protocol in which all the nodes receive $\Tilde{O}(rn^\frac{2}{3})$ qubits in total and has high completeness and low soundness error. 

To complement this result, we claim any $\dMA$ protocol for $\EQ$ with high completeness and low soundness error has to receive $\Omega(rn)$ bits in total by a finer observation of the classical lower bound in \cite{FGNP21}. 

\subsubsection*{Protocol for the greater-than problem and the ranking verification problem}

It was shown that the quantum one-way communication complexity of $\GT$ is maximal, i.e., $\BQP^1(\GT)=\Theta(n)$ by Zhang [Appendix B in \cite{Zha11}]. Therefore, one cannot apply the technique from \cite{FGNP21}, and no efficient $\dQMA$ protocol for $\GT$ was previously known. In this paper, we derive a new way to use quantum fingerprints with classical indexes, and construct an efficient $\dQMA$ protocol for $\GT$.

To construct a $\dQMA$ protocol for the greater-than ($\GT$) problem, we first observe that for $x,y \in \{0,1\}^n$, $\GT(x,y)=1$ if and only if there exists an index $i$ such that a part of $x$ and $y$ from the $1$st bit to the $(i-1)$th bit are the same and the $i$th bit of $x$ is $1$ and the $i$th bit of $y$ is 0. Therefore, we can run the protocol for the equality problem for a part of the inputs, and make the prover send the classical index $i$.

To prove the soundness for the ranking verification problem, we consider to make the prover send a direction bit indicating which input is larger and add a step for a root node to count the directions. We then have an efficient protocol for the ranking verification problem by running the protocol for $\GT$ between multiple terminals in parallel. 

\subsubsection*{Protocol from a quantum one-way communication protocol on general graphs}

To derive a $\dQMA$ protocol for a function that has an efficient quantum one-way communication protocol with multiple terminals, one difficulty is that we need to run the operation of Bob, a party that receives a message from the other party Alice, in the one-way protocol for the function on every leaf. Therefore, we consider a protocol from root to leaves, which is the reverse of the direction of messages in the protocol for $\EQ$. 

The other caveat is that a protocol on one tree is not enough to prove soundness. This is because even if $f(x_i,x_{i+1})=1$ and $f(x_i,x_{i+2})=1$, the value of $f(x_{i+1},x_{i+2})$ can be $0$. To overcome this, we consider running the protocols in parallel for all the $t$ spanning trees whose roots are the $t$ terminals. 

\subsubsection*{Construction of a $\dQMA$ protocol with separable proofs from any $\dQMA$ protocol}

To construct a $\dQMA$ protocol with separable proofs from any $\dQMA$ protocol, we use a $\QMA$ communication complete problem introduced by Raz and Shpilka \cite{RS04}. 

$\QMA$ communication protocols are two-party communication protocols with a prover who can send a proof to one party Alice. Raz and Shpilka \cite{RS04} defined the Linear Subspace Distance (LSD) problem as a $\QMA$ communication complete problem, i.e., any $\QMA$ communication protocol can be reduced to the LSD problem. The LSD problem is a problem to decide whether two subspaces held respectively by the two parties Alice and Bob are close or not.

A useful property of the LSD problem is that it can be solved with a $\QMA$ one-way communication protocol with a proof to Alice. Exploiting this property and the SWAP test strategy \cite{FGNP21}, we construct a $\dQMAsep$ protocol for any function that has a $\QMA$ communication protocol. 

In addition, we observe that any $\dQMA$ protocol can be viewed as a $\QMA$ communication protocol when we split the total nodes into two groups of nodes and consider Alice and Bob to simulate the protocol of the nodes. This leads us to get a $\dQMAsep$ protocol from any $\dQMA$ protocol.

\subsubsection*{Lower bounds for $\dQMA$}

We obtain some lower bounds by counting arguments over quantum states for fooling inputs. To prove our bounds, we use a result from \cite{BCWdW01,dW01}, which states that in order to keep non-trivial distances between each pair of a set of $2^n$ states, at least $\Omega(\log n)$ qubits are required. From this we can prove that, to answer correctly on $2^n$ fooling inputs for $\EQ$ and $\GT$, local nodes in a $\dQMA$ protocol must receive at least $\Omega(\log n)$ qubits. Then, by the pigeonhole principle, we show that at least $\Omega(r \log n)$ qubits are required as a quantum proof in total. This lower bound and proof strategy can be regarded as a quantum analog of the classical lower bound in \cite{FGNP21}.

In order for the above proof strategy to be applicable, proofs between nodes are required to be separable, since entanglement between nodes might fool the verifiers. However, by combining our result on the simulation of any $\dQMA$ protocol with a $\dQMAsep$ protocol, we show a lower bound even for entangled proofs and communications, where the order of the bound is an inverse of a polynomial in $r$, due to the overhead of the simulation. 

For entangled proofs, we can also show a simpler lower bound. Let us suppose that there are consecutive nodes which receive no proof from a prover. Then, even for a function that has only two fooling inputs, the verifiers are easily fooled by the two inputs, because the information the nodes have is separated between the two parts. To deliver quantum proofs to each local node, it can be shown that $\Omega(r)$ qubits are required as a quantum proof in total. By combining the two lower bounds for entangled proofs, for $\EQ$ and $\GT$ we can obtain a lower bound which does not depend on $r$ which is the main result of our lower bounds.

We obtain other lower bounds by a reduction to $\QMA$ communication lower bounds by Klauck \cite{Kla11}. To make a reduction, we first introduce 
$\QMA^*$ communication protocols where proofs are sent to the two parties and they might be entangled. Then, we observe that a $\dQMA$ protocol can be used to give a $\QMA^*$ communication protocol, and then the results of \cite{Kla11} can be applied. 

\subsection{Related works}
Raz and Shpilka \cite{RS04} introduced the Linear Subspace Distance problem as a complete problem for $\QMA$ communication protocols, and showed that there exists an efficient $\QMA$ (two-party) communication protocol and no efficient quantum communication protocol and $\MA$ communication protocol for the problem.
To prove the completeness, they considered a superposition of each step of $\QMA$ communication protocols similar to Kitaev's circuit-to-Hamiltonian construction \cite{KSV02}.

Klauck \cite{Kla11} proved the first lower bounds for the $\QMA$ communication protocols. To derive the lower bounds, Klauck introduced a new technique named one-sided discrepancy, and showed separations between $\mathsf{AM}$ communication complexity and $\mathsf{PP}$ communication complexity, and between $\mathsf{AM}$ communication complexity and $\mathsf{QMA}$ communication complexity. 

In \cite{GMN23a}, Le Gall, Miyamoto, and Nishimura considered the state synthesis \cite{Aar16} on the $\dQMA$ protocols. They introduced the state generation on distributed inputs (SGDI) and gave a $\dQMA$ protocol for the task. As an application, they constructed an efficient $\dQMA$ protocol for the Set Equality problem introduced by \cite{NPY20}. They also showed that from any $\dQMA$ protocol, we can replace quantum communications with classical communications between verifiers on the network and construct an $\mathsf{LOCC}$ (Local Operation and Classical Communication) $\dQMA$ protocol to simulate the original $\dQMA$ protocol.

In \cite{GMN23b}, Le Gall, Miyamoto, and Nishimura introduced distributed quantum interactive proofs ($\mathsf{dQIP}$) as a quantum analog of the distributed interactive proofs ($\mathsf{dAM}$) introduced by \cite{KOS18}. They proved that any $\mathsf{dAM}$ protocols with constant turns communication between verifiers and a prover can be converted into $\mathsf{dQIP}$ protocols with 5 turns if no shared randomness on the network and 3 turns if the shared randomness is allowed.

\subsection{Discussion and open problems}

In this paper, we investigate the power of the $\dQMA$ protocols and show the protocols are indeed useful for many problems but have limits for some functions.

Here we list some problems that are left open by our work.

\begin{enumerate}
    \item There are many variants of $\QMA$ (see \cite{Gha24} for a comprehensive survey on $\QMA$ and its variants) and we can define more variants of $\dQMA$ protocols. For example, we can define a $\mathsf{dQCMA}$ protocol if we allow only classical proofs from a prover while the verifier can communicate with qubits. Another example is a $\dQMA(k)$ protocol for $k \in \mathbb{N}$ if we allow $k$ provers who send quantum proofs to the nodes independently and whose proofs are promised to be separable. Can we find a new relationship between $\dQMA$ protocols and their variants?
    
    Note that some relations are known. In \cite{GMN23a}, the authors showed that any $\dQMA$ protocol can be simulated by an $\mathsf{LOCC}$ $\dQMA$ protocol with some overheads. This paper shows that any $\dQMA$ protocol can be simulated by a $\dQMAsep$ protocol with some overheads.
    \item In our paper and relevant papers about $\dQMA$ protocols, a quantum advantage on the input size is the focus. In \cite{GS16}, G\"{o}\"{o}s and Suomela classified graph properties according to their proof size complexity with local verification based on the graph size. Can we have a quantum advantage in distributed verification concerning the graph size? Can we give an efficient quantum verification protocol for a graph property that is shown to be hard in \cite{GS16}?
    \item There are gaps between upper and lower bounds for $\EQ$ and $\GT$. Can we fill the gaps by providing stronger upper or lower bounds?
\end{enumerate}

\subsection{Organization}
In Section \ref{section:prel}, we give some preliminaries for this paper. In Section \ref{section:perm_eq}, we apply the permutation test to obtain our improved $\dQMA$ protocol for $\EQ$. In Section \ref{section:robust_eq}, we prove a quantum advantage on distributed verification protocols on a path for $\EQ$ still persists even when there is no condition on the size of the path networks. In Section \ref{section:gt}, we derive an efficient $\dQMA$ protocol for $\GT$ and the ranking verification problem. In Section \ref{section:hamming_distance}, we present an efficient $\dQMA$ protocol for the Hamming distance problem with multiple terminals and its applications. In Section \ref{section:dQMAsep}, we show how to convert a $\QMA$ communication protocol and a $\dQMA$ protocol to a $\dQMAsep$ protocol. In Section \ref{section:lower_bound}, we derive some lower bounds for $\dQMA$ protocols.
\section{Preliminaries}\label{section:prel}

When we do not care about constant factors, we use the asymptotic notations. We say $T(n)=O(f(n))$ if there exist constants $c$ and $n_0$ such that for all the integers $n \geq n_0$, we have $T(n) \leq c f(n)$. We say $T(n)=\Omega(f(n))$ if there exist constants $c$ and $n_0$ such that for all the integers $n \geq n_0$, we have $T(n) \geq c f(n)$. $T(n)= \Theta(f(n))$ means that $T(n)=O(f(n))$ and $T(n)=\Omega(f(n))$ hold simultaneously. We also say $T(n)=\Tilde{O}(f(n))$ if there exists a constant $c$ such that $T(n)=O(f(n) \cdot \log^c(f(n)))$. 

This paper considers simple connected graphs as the underlying graph of networks and identifies a network with its underlying graph. The radius $r$ of a network $G=(V,E)$ is defined as
$r := \min_{u\in V}\max_{v\in V}\mathsf{dist}_G(u,v)$, where $\mathsf{dist}_G(u,v)$ denotes the distance between $u$ and $v$ in $G$.

For any event $A$ and $B$, let us denote the complement of $A$ by 
$\neg A$, 
the intersection of $A$ and $B$ by $A \land B$, the union of $A$ and $B$ by $A \lor B$. 
We will need the following basic property on probability.

\begin{lemma}\label{lemma:prob}
    Let $A_j$ be an event for $j=1,2,\ldots,n$.\footnote{Note that these events are not necessarily independent.} Then,
    \[
        \mathrm{Pr}[A_1 \lor A_2 \lor \cdots \lor A_n] \geq \frac{1}{n} \sum_{j=1}^n \mathrm{Pr}[A_j].
    \]
\end{lemma}

\begin{proof}
$n \mathrm{Pr}[A_1 \lor A_2 \lor \cdots \lor A_n] = \sum_{j=1}^n \mathrm{Pr}[A_1 \lor A_2 \lor \cdots \lor A_n] \geq \sum_{j=1}^n \mathrm{Pr}[A_j]$
\end{proof}

\subsection{Quantum computation and information}
We assume that readers are familiar with basic notations of quantum computation and information. We refer to \cite{NC10,Wat18,dW19} for standard references.

For a Hilbert (finite-dimensional complex Euclidean) space $\mathcal{H}$, $\mathcal{B}(\mathcal{H})$ and $\mathcal{D}(\mathcal{H})$ denote the sets of pure and mixed states over $\mathcal{H}$ respectively.
Let us consider Hilbert spaces $\mathcal{H}_1,\ldots,\mathcal{H}_n$ and a matrix $M$ on $\mathcal{H}_1 \otimes \cdots \otimes \mathcal{H}_n$. We will denote by $\ket{b^x_y}$ a $y$th orthonormal basis vector of $\mathcal{H}_x$. Then, let us define the reduced matrix $\text{tr}_{\bar{i}}(M)$ on $\mathcal{H}_i$ obtained by tracing out $\mathcal{H}_1,\ldots,\mathcal{H}_{i-1}$,$\mathcal{H}_{i+1},\ldots,\mathcal{H}_n$ as 
\begin{equation} \nonumber
    \mathrm{tr}_{\bar{i}}(M) = \sum_{j_1,\ldots,j_{i-1},j_{i+1},\ldots,j_n} (
    \bra{b_{j_1}^1} \otimes \cdots \otimes \bra{b_{j_{i-1}}^{i-1}} \otimes I \otimes \bra{b_{j_{i+1}}^{i+1}} \otimes \cdots \otimes \bra{b_{j_{n}}^{n}} ) M (\ket{b_{j_1}^1} \otimes \cdots \otimes \ket{b_{j_{i-1}}^{i-1}} \otimes I \otimes \ket{b_{j_{i+1}}^{i+1}} \otimes \cdots \otimes \ket{b_{j_n}^{n}}).
\end{equation}
We also define the reduced matrix $\text{tr}_{i}(M)$ on $\mathcal{H}_1 \otimes \cdots \otimes \mathcal{H}_{i-1} \otimes \mathcal{H}_{i+1} \otimes \cdots \otimes \mathcal{H}_n$ obtained by tracing out $\mathcal{H}_i$ as 
\begin{equation} \nonumber
    \mathrm{tr}_{i}(M) = \sum_{j} (I \otimes \cdots \otimes I \otimes \bra{b_j^{i}} \otimes I \otimes \cdots \otimes I ) M (I \otimes \cdots \otimes I \otimes \ket{b_j^{i}} \otimes I \otimes \cdots \otimes I ).
\end{equation}

One common measure of distance between quantum states is the trace distance, which is defined as half of the trace norm of the difference of the matrices:
\[
    D(\rho,\sigma) := \frac{1}{2} \|\rho-\sigma\|_1,
\]
where $\|A\|_1 \equiv \mathrm{tr}\sqrt{A^\dag A}$ is the trace norm of $A$, and $\sqrt{A}$ is the unique semidefinite $B$ such that $B^2 = A$ (which is always defined for positive semidefinite $A$). The trace distance can be regarded as a maximum probability to distinguish the two states by POVM measurements since
\[
    D(\rho,\sigma) = \max_M \mathrm{tr}(M (\rho-\sigma)),
\]
where the maximization is taken over all positive operators $M \leq I$. The other common measure of the distance is the fidelity, which is defined as
\[
    F(\rho,\sigma) := \mathrm{tr}\sqrt{\sqrt{\rho}\sigma\sqrt{\rho}}.
\]

The relation between the trace distance and the fidelity is known as follows.
\begin{fact}[Fuchs-van de Graaf inequalities \cite{FvdG99}]\label{fact:Fuchs-van}
    For any quantum states $\rho$ and $\sigma$,
    \[
        1-F(\rho,\sigma) \leq D(\rho,\sigma) \leq \sqrt{1-F(\rho,\sigma)^2}.
    \]
\end{fact}

Here is a useful lemma to connect the trace norm and the fidelity as a corollary of the Uhlmann theorem \cite{Uhl76}.
\begin{lemma}[Corollary 3.23 in \cite{Wat18}]\label{lemma:uhlmann}
    Let $\ket{\psi}$ and $\ket{\phi}$ be two pure states on $\mathcal{X} \otimes \mathcal{Y}$ where $\mathcal{X}$ and $\mathcal{Y}$ are finite-dimensional complex Euclidean spaces. Then,
    \begin{equation}\nonumber
        \| \mathrm{tr}_\mathcal{X} (\ket{\psi} \bra{\phi}) \|_1 = F( \mathrm{tr}_\mathcal{Y}(\ket{\psi}\bra{\psi}), \mathrm{tr}_\mathcal{Y}(\ket{\phi}\bra{\phi}) ).
    \end{equation}
\end{lemma}

We will also need some mathematical facts.
\begin{fact}[Schmidt decomposition, e.g., Theorem 2.7 in \cite{NC10}]\label{fact:Schmidt}
    Suppose $\ket{\psi}$ is a pure state of a composite system $AB$. Then there exist orthonormal states $\ket{i_A} $ for system $A$, and orthonormal states $\ket{i_B}$ of system $B$ such that 
    \[
        \ket{\psi} = \sum_i \lambda_i \ket{i_A} \ket{i_B},
    \]
    where $\lambda_i$ are non-negative numbers satisfying $\sum_i \lambda_i^2 = 1$.
\end{fact}

\begin{fact}\label{fact:distinguishability_quantumalgo}
    For any two mixed states $\rho$ and $\sigma$, any quantum algorithm $\mathcal{A}$ and any classical string s, 
    \begin{equation}\nonumber
        |\mathrm{Pr}[\mathcal{A}(\rho)=s]-\mathrm{Pr}[\mathcal{A}(\sigma)=s]| \leq D(\rho,\sigma).
    \end{equation}
\end{fact}

\begin{fact}\label{fact:tracedistance_contractive}
    The trace distance is contractive under completely positive and trace preserving (CPTP) maps, i.e., if $\Phi$ is a CPTP map, then $D(\Phi(\rho),\Phi(\sigma)) \leq D(\rho,\sigma)$ for any states $\rho$ and $\sigma$. 
\end{fact}

\subsection{Computational models}
In this subsection, we recall definitions of several important computational models and related concepts.

\subsubsection{Communication complexity}\label{subsection:preliminary_communication_complexity}

As standard references, we refer to \cite{KN96,RY20} for classical communication complexity and \cite{dW02, BCMdW10} for quantum communication complexity and the simultaneous message passing (SMP) model.

The goal in communication complexity is for Alice and Bob to compute a function $F : \mathcal{X} \times \mathcal{Y} \to \{0,1,\perp\}$. We interpret $1$ as ``accept'' and $0$ as ``reject'' and we mostly consider $\mathcal{X}=\mathcal{Y}=\{0,1\}^n$. In the computational model, Alice receives an input $x \in \mathcal{X}$ (unknown to Bob) and Bob receives an input $y \in \mathcal{Y}$ (unknown to Alice) promised that $(x,y) \in \mathsf{dom}(F)=F^{-1}(\{0,1\})$. 
In a one-way communication protocol, Alice sends a single message to Bob, and he is required to output $F(x,y)$. In a two-way communication protocol, Alice and Bob can exchange messages with multiple rounds. The cost of a classical (resp.~quantum) communication protocol is the number of bits (resp.~qubits) communicated. 
The (bounded-error) communication complexity (resp.~one-way communication complexity) of $F$ is defined as the minimum cost of two-way (resp.~one-way) classical or quantum communication protocols to compute $F(x,y)$ with high probability, say $\frac{2}{3}$. 

The simultaneous message passing (SMP) model is a specific model of communication protocols. In this model, Alice and Bob each send a single (possibly quantum or randomized) message to a referee Charlie. The goal for Charlie is to output $F(x,y)$ with high probability, say at least $\frac{2}{3}$. The complexity measure of the protocol is the total amount of messages Charlie receives from Alice and Bob.  

In this paper, $\BQP^1(f)$ and $\BQP^{||}(f)$ denote the quantum one-way and SMP communication complexity of $f$, respectively. Note that $\BQP^1(f)\leq \BQP^{||}(f)$ for any $f$ since any SMP protocol can be efficiently simulated by a one-way communication protocol where Charlie is simulated by Bob.

A basic function considered in communication complexity is the equality function $\EQ_n:~\{0,1\}^n\times\{0,1\}^n\rightarrow \{0,1\}$, 
which is defined as $\EQ_n(x,y)=1$ if $x=y$ and $0$ otherwise. This paper frequently uses the fact that $\EQ_n$ can be solved by a one-way quantum protocol of cost $c\log n$ with one-sided error for some constant $c>0$; the protocol outputs $1$ if $x=y$ with probability $1$, and outputs $0$ with probability $2/3$. In what follows, such the protocol is called $\pi$, let $|h_x\rangle$ be the $c\log n$-qubit state from Alice to Bob (fingerprint of $x$), 
and let $\{M_{y,1},M_{y,0}\}$ be the POVM measurement performed by Bob on $\ket{h_x}$, where $M_{y,1}$ corresponds to the measurement result $1$  (accept) and $M_{y,0}$ to the measurement result $0$ (reject).

For any Boolean function $f : \{0,1\}^n \times \{0,1\}^n \to \{0,1\}$, a set $S \subseteq \{0,1\}^n \times \{0,1\}^n$ is a 1-fooling set for $f$ if $f(x,y)=1$ for any $(x,y) \in S$ ,and $f(x_1,y_2)=0$ or $f(x_2,y_1)=0$ for any two pairs $(x_1,y_1) \neq (x_2,y_2) \in S \times S$. 

\subsubsection{$\QMA$ communication protocols and its variants}\label{subsubsection:QMA_communication}

Let us recall the definition of $\QMA$ communication protocols.
\begin{definition}[$\QMA$ communication protocol and $\QMAcc(f)$, Definition 3 in \cite{Kla11} and Definition 4 in \cite{RS04}]
    In a $\QMA$ communication protocol for an input $(x,y)$, Alice has a part of the input $x$ and Bob has the other part of the input $y$, and Merlin produces a quantum state $\rho$ (the proof) on some $\gamma$ qubits, which he sends to Alice. Alice and Bob then communicate using a quantum protocol of $\mu$ qubits in total with multiple rounds, and either accept or reject the input $(x, y)$. We say that a $\QMA$ communication protocol computes a Boolean function $f: \{0,1\}^n \times \{0,1\}^n \rightarrow \{0,1\}$, if for all inputs $(x,y)$ such that $f(x, y) = 1$, there exists a quantum proof such that the protocol accepts with probability at least $\frac{2}{3}$, and for all inputs $(x,y)$ such that $f(x,y)=0$, and all quantum proofs, the protocol accepts with probability at most $\frac{1}{3}$. 
    The cost of a $\QMA$ communication protocol is the sum of the proof size $\gamma$ and the length of the communication $\mu$ between Alice and Bob. We define $\QMAcc(f)$ as the minimum cost of the protocol that computes $f$.
\end{definition}

We say for a function $f$, $\QMAcc(f)=\gamma+\mu$ if there exists a $\QMA$ communication protocol whose proof size is $\gamma$ and communication amount is $\mu$.

Next, let us define a $\QMA$ one-way communication protocol and a $\QMA^*$ communication protocol as two variants of the $\QMA$ communication protocol. In the $\QMA$ one-way communication protocol, Alice can send a message once to Bob and no more communication is prohibited.

\begin{definition}[$\QMA$ one-way communication protocol and $\QMAcc^1(f)$]
    In a $\QMA$ one-way communication protocol for an input $(x,y)$, Alice has a part of the input $x$ and Bob has the other part of the input $y$, and Merlin produces a quantum state $\rho$ (the proof) on some $\gamma$ qubits, which he sends to Alice. Alice applies some quantum operations on the proof depending on her input $x$ and sends $\mu$ qubits to Bob. Bob applies some quantum operations depending on his input $y$ and outputs accept or reject. We say that a $\QMA$ one-way communication protocol computes a Boolean function $f: \{0,1\}^n \times \{0,1\}^n \rightarrow \{0,1\}$, if for all inputs $(x,y)$ such that $f(x, y) = 1$, there exists a quantum proof such that the protocol accepts with probability at least $\frac{2}{3}$, and for all inputs $(x,y)$ such that $f(x,y)=0$, the protocol accepts with probability at most $\frac{1}{3}$ for any quantum proof. 
    The cost of a $\QMA$ one-way communication protocol is the sum of the proof size $\gamma$ and the length of the one-way communication $\mu$ from Alice to Bob. We define $\QMAcc^1(f)$ as the minimum cost of the protocol that computes f.
\end{definition}

In the $\QMA^*$ communication protocol, Alice and Bob can receive proofs respectively from Merlin and the proofs might be entangled.

\begin{definition}[$\QMA^*$ communication protocol and $\QMAcc^* (f)$]
    In a $\QMA^*$ communication protocol for an input $(x,y)$, Alice has a part of the input $x$ and Bob has the other part of the input $y$, and Merlin produces a quantum state $\rho$ (the proof) on some $(\gamma_1+\gamma_2)$ qubits, which he sends $\gamma_1$ qubits to Alice and $\gamma_2$ qubits to Bob. Alice and Bob then communicate using a quantum protocol of $\mu$ qubits in total with multiple rounds, and either accept or reject the input $(x,y)$. We say that a $\QMA^*$ communication protocol computes a Boolean function $f: \{0,1\}^n \times \{0,1\}^n \rightarrow \{0,1\}$, if for all inputs $(x,y)$ such that $f(x, y) = 1$, there exists a quantum proof such that the protocol accepts with probability at least $\frac{2}{3}$, and for all inputs $(x,y)$ such that $f(x,y)=0$, and all quantum proofs, the protocol accepts with probability at most $\frac{1}{3}$. 
    The cost of a $\QMA^*$ communication protocol is the sum of the total proof size $\gamma_1+\gamma_2$ and the length of the communication $\mu$ between Alice and Bob. We define $\QMAcc^*(f)$ as the minimum cost of the protocol that computes $f$ on all the inputs.
\end{definition}

We say for a function $f$, $\QMAcc^1(f)=\gamma+\mu$ and $\QMAcc^*(f)=\gamma_1+\gamma_2+\mu$ similar to $\QMAcc(f)=\gamma+\mu$.
There are (trivial) relationships between them. First, for any $f$, $\QMAcc(f) \leq \QMAcc^1(f)$ by their definitions. Second, for any $f$ for which $\QMAcc^*(f)=\gamma_1+\gamma_2+\mu$, 
\begin{equation}\label{equation:relation}
    \QMAcc(f) \leq \gamma_1 + 2 \gamma_2 + \mu.
\end{equation}
This is because 
any $\QMA^*$ communication protocol $\mathcal{P}$
such that Merlin sends $\gamma_1$ qubits to Alice and $\gamma_2$ qubits to Bob can be simulated by a $\QMA$ communication protocol where Merlin sends Alice the $(\gamma_1+\gamma_2)$ qubits sent by Merlin in $\mathcal{P}$, Alice sends the Bob-part in $\mathcal{P}$ ($\gamma_2$ qubits) to Bob, and Alice and Bob conduct the subsequent communication protocol in $\mathcal{P}$. 

\subsubsection{Distributed verification}

Let us recall the definition of classical distributed verification protocols called distributed Merlin-Arthur protocols ($\dMA$ protocols). 

In a $\nu$-round $\dMA$ protocol for a binary-valued function $f$, the prover (Merlin) first sends a message called a proof (or certificate) to the verifier (Arthur) that consists of the nodes of a network $G=(V,E)$. More precisely, the prover sends a $c(u)$-bit string to each $u\in V$. 
Then, the nodes of $G$ run a $\nu$-round verification algorithm, namely, a randomized algorithm (or protocol) using $\nu$-round communication among the nodes. 
Here, $t$ nodes $u_i$ called {\em terminals} have own input string $x_i$. 
Then, the condition that the $\dMA$ protocol should satisfy for verifying whether $f(x_1,\ldots,x_t)=1$ or not is as follows. 

\begin{definition}\label{definition:dMA}
    On a network $G=(V,E)$, a $\nu$-round $\dMA$ protocol $\pi$ of $c(u)$ bits proof for $u \in V$ and $m(v,w)$ bits communication for $\{v,w\} \in E$ has completeness $a$ and soundness $b$ for a function $f:(\{0,1\}^n)^t \to \{0,1\}$ if there exists a $\nu$-round 
    verification algorithm with messages of $m(u,v)$ bits in total between nodes $v$ and $w$ for $\{v,w\} \in E$ respectively such that for all the inputs $(x_1,\dots,x_t)\in (\{0,1\}^n)^t:$
    \begin{itemize}
	   \item \textbf{Completeness:} if $f(x_1,\dots,x_t)=1$, then there exists a $(\sum_{u \in V} c(u))$-bit proof to the nodes such that $\Pr[\mbox{all the nodes accept}]\geq a;$ 
	   \item \textbf{Soundness:} if $f(x_1,\dots,x_t)=0$, then $\Pr[\mbox{all the nodes accept}]\leq b$ for any $(\sum_{u \in V} c(u))$-bit proof.  
    \end{itemize}
    \sloppy In particular, we say that the protocol $\pi$ has perfect completeness if $a=1$. The sum $\sum_{u\in V}c(v)$ (resp.~$\sum_{\{v,w\} \in E} m(v,w)$) is called the total proof (resp.~message) size of $\pi$, and $\max_{u\in V}c(v)$ (resp.~$\max_{\{u,w\} \in E} m(v,w)$) is called the local proof (resp.~message) size of $\pi$. 
\end{definition}

Let us next recall the definition of quantum verification protocols called distributed quantum Merlin-Arthur protocols ($\dQMA$ protocols). 
A $\dQMA$ protocol is defined similarly to $\dMA$ protocols 
except that the message from the prover is a quantum state 
and the algorithm of each node and the communication among the nodes are also quantum (and thus the complexity is measured by the number of qubits). 
The condition that the $\dQMA$ protocol should satisfy 
for verifying whether $f(x_1,\ldots,x_n)=1$ or not is as follows; let $\mathcal{H}_v$ denote the Hilbert space associated with the quantum register $R_v$ sent from the prover to the node $v$.

\begin{definition}\label{definition:dQMA}
    On a network $G=(V,E)$, a $\nu$-round $\dQMA$ protocol of $c(u)$ qubits proof for $u \in V$ and $m(v,w)$ qubits communication for $\{v,w\} \in E$ has completeness $a$ and soundness $b$ for a function $f:(\{0,1\}^n)^t \to \{0,1\}$ if there exists a $\nu$-round quantum verification algorithm with messages of $m(v,w)$ qubits in total between nodes $v$ and $w$ for $\{v,w\} \in E$ respectively such that for all the inputs $(x_1,\dots,x_t)\in (\{0,1\}^n)^t:$
    \begin{itemize}
	   \item \textbf{Completeness:} if $f(x_1,\dots,x_t)=1$, then there exists a $(\sum_{u \in V} c(u))$-qubit proof $\ket{\xi}$ on the Hilbert space $\bigotimes_{u \in V} \mathcal{H}_u$ to the nodes such that $\Pr[\mbox{all the nodes accept}]\geq a;$
	   \item \textbf{Soundness:} if $f(x_1,\dots,x_t)=0$, then for any $(\sum_{u \in V} c(u))$-qubit proof $\ket{\xi}$ on $\bigotimes_{u \in V} \mathcal{H}_u$, $\Pr[\mbox{all the nodes accept}]\leq b$.  
    \end{itemize}
\end{definition}

In the definition above, we consider quantum proofs that are only pure states. Since mixed states are convex combinations of pure states, this restriction does not affect the completeness and soundness parameters and lose generality as in the case for $\QMA$.

Let us define some variants of the $\dQMA$ protocol. For $\dQMAsep$ protocols, the completeness holds with a separable proof between nodes and the soundness holds against any entangled proof. Actually, the $\dQMA$ protocols in \cite{FGNP21} as well as the $\dQMA$ protocols in this paper are $\dQMAsep$ protocols while we do not state it in some of the theorems for the simplicity of their statements. 

\begin{definition}\label{definition:dQMAsep}
    On a network $G=(V,E)$, a $\nu$-round $\dQMA^\mathsf{sep}$ protocol of $c(u)$ qubits proof for $u \in V$ and $m(v,w)$ qubits communication for $\{v,w\} \in E$ has completeness $a$ and soundness $b$ for a function $f:(\{0,1\}^n)^t \to \{0,1\}$ if there exists a $\nu$-round quantum verification algorithm with messages of $m(v,w)$ qubits between nodes $v$ and $w$ for $\{v,w\} \in E$ respectively such that for all the inputs $(x_1,\dots,x_t)\in (\{0,1\}^n)^t:$
    \begin{itemize}
	   \item \textbf{Completeness:} if $f(x_1,\dots,x_t)=1$, then there is a $(\sum_{u \in V} c(u))$-qubit proof $\bigotimes_{u \in V} \ket{\xi_u}$, where $\ket{\xi_u}$ is a state on $\mathcal{H}_u$ for $u \in V$,  to the nodes such that $\Pr[\mbox{all the nodes accept}]\geq a;$
	   \item \textbf{Soundness:} if $f(x_1,\dots,x_t)=0$, then for any $(\sum_{u \in V} c(u))$-qubit proof $\ket{\xi}$ on $\bigotimes_{u \in V} \mathcal{H}_u$, $\Pr[\mbox{all the nodes accept}]\leq b$.  
    \end{itemize}
\end{definition}

For $\dQMAsepsep$ protocols, the completeness holds with a separable proof and the soundness holds against only separable proofs. In other words, a $\dQMAsepsep$ protocol is a $\dQMA$ protocol where a prover can send only separable proofs over nodes.

\begin{definition}\label{definition:dQMAsepsep}
    On a network $G=(V,E)$, a $\nu$-round $\dQMAsepsep$ protocol of $c(u)$ qubits proof for $u \in V$ and $m(v,w)$ qubits communication for $\{v,w\} \in E$ has completeness $a$ and soundness $b$ for a function $f:(\{0,1\}^n)^t \to \{0,1\}$ if there exists a $\nu$-round quantum verification algorithm with messages of $m(v,w)$ qubits between nodes $v$ and $w$ for $\{v,w\} \in E$ respectively such that for all the inputs $(x_1,\dots,x_t)\in (\{0,1\}^n)^t:$
    \begin{itemize}
	   \item \textbf{Completeness:} if $f(x_1,\dots,x_t)=1$, then there is a $(\sum_{u \in V} c(u))$-qubit proof $\bigotimes_{u \in V} \ket{\xi_u}$, where $\ket{\xi_u}$ is a state on $\mathcal{H}_u$ for $u \in V$,  to the nodes such that $\Pr[\mbox{all the nodes accept}]\geq a;$
	   \item \textbf{Soundness:} if $f(x_1,\dots,x_t)=0$, then for any $(\sum_{u \in V} c(u))$-qubit proof $\bigotimes_{u \in V} \ket{\xi_u}$, where $\ket{\xi_u}$ is a state on $\mathcal{H}_u$ for $u \in V$, $\Pr[\mbox{all the nodes accept}]\leq b$.  
    \end{itemize}
\end{definition}

Note that if a protocol $\mathcal{P}$ is a $\dQMAsep$ protocol, then $\mathcal{P}$ is also a $\dQMAsepsep$ protocol from the definitions.

In what follows, a distributed verification protocol is a $1$-round one when we do not mention the number of rounds explicitly.

\section{Improved $\dQMA$ protocol for $\EQ$ with the permutation test}\label{section:perm_eq}

In this section, we derive a $\dQMA$ protocol for the equality function exploiting the property of the permutation test.

\subsection{Property and application of the permutation test}\label{subsection:perm}
The permutation test \cite{BBD+97,BCWdW01,KNY08} is a generalization of the SWAP test. In this subsection, we identify the property of the permutation test as a special case of the weak Schur sampling and the generalized phase estimation \cite{Har05, CHW07}. We refer to Section 4.2.2 in \cite{MW16} for a comprehensive summary. We then apply the property of the permutation test to check how the reduced states on subsystems are close.

First, let us recall the SWAP test. The test is a protocol with a given input state on $\mathcal{H}=\mathcal{H}_1\otimes \mathcal{H}_2$ where $\mathcal{H}_1$ and $\mathcal{H}_2$ are Hilbert spaces. We here consider $\mathcal{H}_1$ and $\mathcal{H}_2$ are corresponding to registers $R_1$ and $R_2$.
\begin{algorithm}[H]
\caption{\,The SWAP test}
\textbf{Input:\,} $\rho \in \mathcal{D}(\mathcal{H}_1 \otimes \mathcal{H}_2)$ on registers $R_1$ and $R_2$.
\begin{algorithmic}[1]
\State  Prepare an ancilla qubit and initialize the state with $\ket{0}$.
\State Apply the Hadamard gate $H=\frac{1}{\sqrt{2}}
\left(
\begin{matrix}
1 & 1\\
1 & -1
\end{matrix}
\right)$ on the state and obtain the state $\ket{+}=\frac{1}{\sqrt{2}}(\ket{0}+\ket{1})$.
\State Apply the controlled swap $\ket{0}\bra{0} \otimes I + \ket{1}\bra{1} \otimes \mathrm{SWAP}$ where $\mathrm{SWAP}$ is defined by $\mathrm{SWAP} \ket{i_1}\ket{i_2} = \ket{i_2}\ket{i_1}$ for $\ket{i_1} \in \mathcal{B}(\mathcal{H}_1)$ and $\ket{i_2}\in  \mathcal{B}(\mathcal{H}_2)$.
\State Apply the Hadamard gate again on the ancilla qubit and measure it in the computational basis. If the measurement result is $\ket{0}$, the test accepts. Else, it rejects.
\end{algorithmic}
\end{algorithm}

It is well known that when pure states $|\psi_1\rangle$ and $|\psi_2\rangle$ on $R_1$ and $R_2$ are given, the SWAP test accepts with probability 
$\frac{1}{2}+\frac{1}{2}|\langle\psi_1|\psi_2\rangle|^2$. 
In particular, the SWAP test accepts with probability $1$ 
when $|\psi_1\rangle=|\psi_2\rangle$. 

For completeness, we rewrite the lemmas about the property and application of the SWAP test from \cite{FGNP21}, which will be used in Section \ref{subsection:path}. Let $\mathcal{H}_S^2$ denote the symmetric subspace of $\mathcal{H}_1 \otimes \mathcal{H}_2$ and let $\mathcal{H}_A$ denote the anti-symmetric subspace in $\mathcal{H}_1 \otimes \mathcal{H}_2$. Note that any state in $\mathcal{H}_1 \otimes \mathcal{H}_2$ can be represented as a superposition of a state in $\mathcal{H}_S^2$ and a state in $\mathcal{H}_A$, i.e., $\mathcal{H}_1 \otimes \mathcal{H}_2 = \mathcal{H}_S^2 \oplus \mathcal{H}_A$ since SWAP is a Hermitian matrix which has only $+1$ and $-1$ eigenvalues.

\begin{lemma}[Lemma 4 in \cite{FGNP21}]\label{lemma:swap_ac}
	Assume that $|\psi\rangle=\alpha|\psi_S\rangle+\beta|\psi_A\rangle$ 
	where $|\psi_S\rangle \in \mathcal{B}(\mathcal{H}_S^2)$ and $|\psi_A\rangle \in \mathcal{B}(\mathcal{H}_A)$.
	Then, the SWAP test on input $|\psi\rangle$ accepts with probability $|\alpha|^2$. 
\end{lemma}

\begin{lemma}[Lemma 5 in \cite{FGNP21}]\label{lemma:swap_close}
	Let $0 \leq \epsilon \leq 1$, and assume that the SWAP test on input $\rho$ in the input register $(R_1,R_2)$ 
	accepts with probability $1-\epsilon$.
	Then, $D(\rho_1, \rho_2 ) \leq 2\sqrt{\epsilon}+ \epsilon$, 
	where $\rho_j$ is the reduced state on $R_j$ of $\rho$.
	Moreover, if the SWAP test on input $\rho$ accepts with probability 1, 
	then $\rho_1=\rho_2$ (and hence $D(\rho_1, \rho_2)=0$).
\end{lemma}

The SWAP test can be considered as a test to estimate the absolute value of the amplitude in the symmetric subspace of a bipartite system. 
We will next generalize the test to $k$-partite systems for any integer $k$. Let $S_k$ denote the symmetric group on $k$ elements and define a unitary operator $U_\pi$ which acts by permuting $k$-partite systems according to $\pi$ as
\[
    U_\pi \ket{i_1}\cdots\ket{i_k} = \ket{i_{\pi^{-1}(1)}}\cdots\ket{i_{\pi^{-1}(k)}}.
\]
Let $\lambda$ denote a partition of $\{1,\ldots,k\}$ that corresponds to an irreducible representation (irrep) of $S_k$. We denote $d_\lambda$ the dimension of the corresponding irreducible representation $V_\lambda$ of $S_k$, which associates a $d_\lambda$-dimensional square matrix with each permutation $\pi \in S_k$. The quantum Fourier transform (QFT) over $S_k$ is a unitary operator that performs a change of bases from $\{ \ket{\pi}: \pi \in S_k \}$ to $\{ \ket{\lambda,i,j}: 1 \leq i,j \leq d_\lambda \}$. Then, the algorithm of the permutation test can be described as Algorithm \ref{algorithm:permutation_test}.

\begin{algorithm}[H]
\caption{\,The permutation test}
\textbf{Input:\,} $\rho \in \mathcal{D}(\mathcal{H}_1 \otimes \cdots \otimes \mathcal{H}_k)$ on registers $R_1,\ldots,R_k$.
\label{algorithm:permutation_test}
\begin{algorithmic}[1]
\State Prepare a $(k!)$-dimensional ancilla register whose basis states correspond to $\ket{\lambda,i,j}$.
\State Initialize the ancilla register in the state $\ket{(k),1,1}$ where $(k)$ is corresponding to the trivial irrep.
\State Apply the inverse quantum Fourier transform over $S_k$ to the ancilla qubits and obtain the state $\frac{1}{\sqrt{k!}} \sum_{\pi \in S_k} \ket{\pi}$. 
\State Apply the controlled permutation $\sum_{\pi \in S_k} \ket{\pi}\bra{\pi} \otimes U_\pi$.
\State Apply the quantum Fourier transform over $S_k$ to the ancilla and measure it in the computational basis.
\State If the measurement result of the partition $\lambda$ is $(k)$, the test accepts. Else, it rejects.
\end{algorithmic}
\end{algorithm}
The probability that $\lambda$ is output is $\text{tr}(P_\lambda\rho)$ \cite{BCH06,Har05}. The projector $P_\lambda$ is defined by
\begin{equation}\nonumber
    P_\lambda := \frac{d_\lambda}{k!}\sum_{\pi \in S_k} \chi_\lambda(\pi) U_\pi,
\end{equation}
where $\chi_\lambda$ is the corresponding character $\mathrm{tr}(V_\lambda)$. In this paper, we concentrate on the case where $\lambda$ is the trivial irrep ($k$) which maps $\pi \mapsto 1$ for all $\pi \in S_k$. In the case, $d_\lambda = 1$ and $\chi_\lambda(\pi) = 1$ for all $\pi \in S_k$. Therefore, $P_\lambda = \frac{1}{k!} \sum_{\pi \in S_k} U_\pi$. This is equal to $\displaystyle { d+k-1 \choose k } \int d \psi \ket{\psi}^{\otimes k} \bra{\psi}^{\otimes k}$, which is the projector $\Pi_{\mathrm{sym}}$ to the symmetric subspace $\mathcal{H}_S^k := \{ \ket{\Phi} \in \mathcal{B}((\mathbb{C}^d)^{\otimes k}):U_\pi \ket{\Phi} = \ket{\Phi}\}$. 
See e.g., Lemma 1.7 in \cite{chr06} and Lemma 1 in \cite{Sco06} for the reference of this fact.

The following lemma is an analog of Lemma \ref{lemma:swap_ac} for the $k$-partite case using the permutation test. We will denote by $\mathcal{H}_N$ the orthogonal subspace of $\mathcal{H}_1 \otimes \cdot \cdot \cdot \otimes \mathcal{H}_k$ to the symmetric subspace $\mathcal{H}_S^k$, i.e., $\mathcal{H}_1 \otimes \cdot \cdot \cdot \otimes \mathcal{H}_k = \mathcal{H}_S^k \oplus \mathcal{H}_N$. 
\begin{lemma}\label{lemma:perm_ac}
Assume that $\ket{\psi} = \Pi_{\mathrm{sym}} (\ket{\psi}) + (I - \Pi_{\mathrm{sym}}) (\ket{\psi}) = \alpha \ket{\psi_S} + \beta \ket{\psi_N}$ where $\ket{\psi_S} \in \mathcal{B}(\mathcal{H}_S^k)$ and $\ket{\psi_N} \in \mathcal{B}(\mathcal{H}_N)$. Then, the permutation test on input $\ket{\psi}$ accepts with probability $|\alpha|^2$. 
In particular, the test accepts with probability $1$ if $|\psi\rangle=|\varphi\rangle^{\otimes k}$ for some $|\varphi\rangle$. 
\end{lemma}
The following lemma is also an analog of Lemma \ref{lemma:swap_close} for the $k$-partite case using the permutation test. Note that a similar analysis was first done by Rosgen (Lemma 5.1 in \cite{Ros08}) with the fidelity as a measure between quantum states.
\begin{lemma}\label{lemma:perm_close}
Let $0 \leq \epsilon \leq 1$, and assume the permutation test on input $\rho$ in the registers $R_1,\ldots,R_n$ accepts with probability $1-\epsilon$. Then, for any $i,j \in [n]$, $D(\rho_i,\rho_j) \leq 2\sqrt{\epsilon} + \epsilon$ where $\rho_i$ and $\rho_j$ are the reduced states of $R_i$ and $R_j$ respectively. Moreover, if the permutation test on input $\rho$ accepts with probability 1, then, for any $i,j \in [n]$, $\rho_i = \rho_j$ (and hence $D(\rho_i,\rho_j)=0$).
\end{lemma}

\begin{proof}
The mixed state $\rho$ can be decomposed into an ensemble of pure states as $\sum_k p_k \ket{\psi_k} \bra{\psi_k}$. In addition, each pure state is a superposition of a state in the symmetric subspace and a state in the orthogonal subspace, namely $\ket{\psi_k} = \alpha_k \ket{\psi_k^S} + \beta_k \ket{\psi_k^N}$. By Lemma \ref{lemma:perm_ac} and the assumption of the acceptance probability, $\sum_k p_k |\alpha_k|^2 \geq 1-\epsilon$. Then, 
\begin{equation}\label{equation:sum_beta}
    \sum_k p_k |\beta_k|^2 \leq \epsilon.
\end{equation}
The state $\rho$ can be moreover represented as
\begin{equation}\nonumber
    \rho = \sum_k p_k ( |\alpha_k|^2 \ket{\psi_k^S} \bra{\psi_k^S} + \alpha_k \beta_k^* \ket{\psi_k^S} \bra{\psi_k^N} + \alpha_k^*\beta_k \ket{\psi_k^N} \bra{\psi_k^S} + |\beta_k|^2 \ket{\psi_k^N} \bra{\psi_k^N} ).
\end{equation}
Let us denote $\psi_k^{s} = \ket{\psi_k^S} \bra{\psi_k^S}$, $\psi_k^{sn} = \ket{\psi_k^S} \bra{\psi_k^N}$, $\psi_k^{ns} = \ket{\psi_k^N} \bra{\psi_k^S}$ and $\psi_k^{n} = \ket{\psi_k^N} \bra{\psi_k^N}$. Using the notations, the subsystems $\rho_i$ and $\rho_j$ can be described as follows.
\begin{eqnarray}
    \rho_i = \sum_k p_k ( |\alpha_k|^2 \text{tr}_{\bar{i}}(\psi_k^{s}) + \alpha_k \beta_k^* \text{tr}_{\bar{i}}(\psi_k^{sn}) + \alpha_k^*\beta_k \text{tr}_{\bar{i}}(\psi_k^{ns}) + |\beta_k|^2 \text{tr}_{\bar{i}}(\psi_k^{n}) ), \nonumber \\
    \rho_j = \sum_k p_k ( |\alpha_k|^2 \text{tr}_{\bar{j}}(\psi_k^{s}) + \alpha_k \beta_k^* \text{tr}_{\bar{j}}(\psi_k^{sn}) + \alpha_k^*\beta_k \text{tr}_{\bar{j}}(\psi_k^{ns}) + |\beta_k|^2 \text{tr}_{\bar{j}}(\psi_k^{n}) ). \nonumber
\end{eqnarray}
From the definition of the symmetric subspace, $\text{tr}_{\bar{i}}(\psi_k^{s}) = \text{tr}_{\bar{j}}(\psi_k^{s})$. We then get
\begin{equation}\nonumber
    \rho_i - \rho_j = \sum_k p_k ( \alpha_k \beta_k^* (\text{tr}_{\bar{i}}(\psi_k^{sn}) - \text{tr}_{\bar{j}}(\psi_k^{sn})) + \alpha_k^* \beta_k (\text{tr}_{\bar{i}}(\psi_k^{ns}) - \text{tr}_{\bar{j}}(\psi_k^{ns})) + |\beta_k|^2 (\text{tr}_{\bar{i}}(\psi_k^{n}) - \text{tr}_{\bar{j}}(\psi_k^{n})) ).
\end{equation}
From the positive scalability and the triangle inequality of the trace norm, we obtain
\begin{eqnarray}
    \lefteqn{ 
    D(\rho_i,\rho_j) = \frac{1}{2} \| \rho_i - \rho_j \|_1 } \nonumber \\ 
    &\leq& \frac{1}{2} \sum_k p_k ( |\alpha_k| |\beta_k| \| \text{tr}_{\bar{i}}(\psi_k^{sn}) - \text{tr}_{\bar{j}}(\psi_k^{sn}) \|_1 + |\alpha_k| |\beta_k| \| \text{tr}_{\bar{i}} (\psi_k^{ns}) - \text{tr}_{\bar{j}}(\psi_k^{ns}) \|_1 + |\beta_k|^2 \| \text{tr}_{\bar{i}}(\psi_k^{n}) - \text{tr}_{\bar{j}}(\psi_k^{n})\|_1 ). \nonumber
\end{eqnarray}
Since $\text{tr}_{\bar{i}}(\psi_k^{n})$ and $\text{tr}_{\bar{j}}(\psi_k^{n})$ are quantum states, their trace norms are 1. We thus have
\begin{equation}\nonumber
    \| \text{tr}_{\bar{i}}(\psi_k^{n}) - \text{tr}_{\bar{j}}(\psi_k^{n}) \|_1 \leq \| \text{tr}_{\bar{i}}(\psi_k^{n}) \|_1 + \| \text{tr}_{\bar{j}}(\psi_k^{n}) \|_1 = 1+1 = 2.
\end{equation}
With Lemma \ref{lemma:uhlmann} and the fact that the fidelity between any quantum states can be bounded by 1,
\begin{eqnarray}
     \| \text{tr}_{\bar{i}}(\psi_k^{sn}) \|_1 = F(\text{tr}_i(\psi_k^s), \text{tr}_i(\psi_k^n)) \leq 1, \nonumber \\
      \| \text{tr}_{\bar{j}}(\psi_k^{sn}) \|_1 = F(\text{tr}_j(\psi_k^s), \text{tr}_j(\psi_k^n)) \leq 1. \nonumber
\end{eqnarray}
We hence have
\begin{equation}\nonumber
    \| \text{tr}_{\bar{i}}(\psi_k^{sn}) - \text{tr}_{\bar{j}}(\psi_k^{sn}) \|_1 \leq \| \text{tr}_{\bar{i}}(\psi_k^{sn}) \|_1 + \| \text{tr}_{\bar{j}}(\psi_k^{sn}) \|_1 = 1+1 = 2.
\end{equation}
A similar argument holds as $\| \text{tr}_{\bar{i}}(\psi_k^{ns}) - \text{tr}_{\bar{j}}(\psi_k^{ns}) \|_1 \leq 2.$ Therefore, we have
\begin{equation}\nonumber
    D(\rho_i,\rho_j) \leq \sum_k p_k 2 |\alpha_k||\beta_k| + \sum_k p_k |\beta_k|^2. 
\end{equation}
From Eq. (\ref{equation:sum_beta}), the Cauchy-Schwarz inequality and $|\alpha_k| \leq 1$,
\begin{eqnarray}
    \sum_k p_k 2 |\alpha_k||\beta_k| + \sum_j p_k |\beta_k|^2 &\leq& 2 \sum_k p_k |\beta_k| + \epsilon \nonumber \\
    &=& 2 \sum_k \sqrt{p_k} \sqrt{p_k} |\beta_k| + \epsilon \nonumber \\
    &\leq& 2 \left(\sum_k p_k\right)^{\frac{1}{2}} \left(\sum_k p_k|\beta_k|^2\right)^\frac{1}{2} + \epsilon \nonumber \\
    &\leq& 2 \sqrt{\epsilon} + \epsilon, \nonumber
\end{eqnarray}
which concludes the proof.
\end{proof}

\subsection{Protocol on paths}\label{subsection:path}

In this subsection, we focus on the case where the verifier $v_0,\ldots,v_r$ are arranged in a row and the two extremities $v_0$ and $v_r$ have inputs. Let $x \in \{0,1\}^n$ be the input string owned by $v_0$, 
and $y \in \{0,1\}^n$ be the input string owned by $v_r$. We are going to derive a $\dQMA$ protocol for the equality function $\EQ$. 

Our $\dQMA$ protocol $\mathcal{P}_\pi$ is described in Algorithm \ref{algorithm:path} (recall that $\pi$, $|h_x\rangle$, and $\{M_{y,1},M_{y,0}\}$ are defined in Section \ref{subsection:preliminary_communication_complexity}). 

\begin{algorithm}[H]
\caption{\, Protocol $\mathcal{P}_\pi$ for $\EQ$ on an input pair $(x,y)$ in a path $v_0,\ldots,v_r$}
\label{algorithm:path}
\begin{algorithmic}[1]
\State The prover sends two $c\log n$-qubit registers $R_{j,0}, R_{j,1}$ to each of the intermediate nodes $v_j$ for $j \in \{1,\ldots,r-1\}$. 
\State The left-end node $v_0$ prepares the state $\rho_0 = \ket{h_x}\bra{h_x}$ in the register $R_0$ by itself, 
and sends $R_0$ to the right neighbor $v_1$.
\State Each intermediate node $v_j$ swaps the states between $R_{j,0}$ and $R_{j,1}$ with probability $\frac{1}{2}$, i.e., symmetrizes the states on $R_{j,0}$ and $R_{j,1}$. \label{step:path_symmetrize}
\State Each intermediate node $v_j$ sends $R_{j,1}$ to the right neighbor $v_{j+1}$.
\State Each intermediate node $v_j$ receives $R_{j-1,1}$ from its left neighbor $v_{j-1}$. Then $v_j$ performs the SWAP test on the registers $(R_{{j-1},1}, R_{j,0})$ and accepts or rejects accordingly. \label{step:path_swap}
\State The right-end node $v_r$ receives $R_{r-1,1}$ from its left neighbor $v_{r-1}$. Then, $v_r$ performs the POVM measurement $\{M_{y,0},M_{y,1}\}$ corresponding to $\pi$ applied to the state $R_{r-1,1}$ and accepts or rejects accordingly. 
\end{algorithmic}
\end{algorithm}

In the above protocol $\mathcal{P}_\pi$, the size of the quantum proof that each node receives from the prover is $2c \log n$,
and the length of the quantum message that each node sends to the neighbor is $c \log n$.  
We next show that the above protocol has perfect completeness and soundness $\frac{4}{81r^2}$.

\subsubsection*{Completeness}

Let us assume inputs $x$ and $y$ are satisfying $\EQ(x,y)=1$, i.e., $x=y$. The prover sends $\ket{h_x} \ket{h_x}$ to all the intermediate nodes. In step \ref{step:path_symmetrize}, as the state is already symmetric, the state does not change by the symmetrization. Therefore, in step \ref{step:path_swap}, all the SWAP tests accept with certainty. Furthermore, the right end node $v_r$ accepts with certainty. Then, from the definition of completeness, the protocol has perfect completeness.

\subsubsection*{Soundness}
Let us assume inputs $x$ and $y$ are satisfying $\mathsf{EQ}(x,y)=0$, i.e., $x \neq y$. Then, the following lemma holds.
\begin{lemma}\label{lemma:sum_rej}
    For $j \in \{1,\ldots,r\}$, let $E_j$ be the event that the local test $v_j$ performs (the SWAP test or the POVM measurement) accepts. Then, $\sum_{j=1}^r \mathrm{Pr}[\neg{E_j}] \geq \frac{4}{81r}$.
\end{lemma}
\begin{proof}
    For conciseness, let us denote $p_j = \mathrm{Pr}[\neg{E_j}]$. By Lemma \ref{lemma:swap_close}, the trace distance between the reduced states $\rho_{j-1,1}$ on $R_{j-1,1}$ and $\rho_{j,0}$ on $R_{j,0}$ can be bounded as
    \[
        D( \rho_{j-1,1},\rho_{j,0} ) \leq 2\sqrt{p_j} + p_j.
    \]
    We thus have $D( \rho_{j-1,1},\rho_{j,0} ) \leq 3 \sqrt{p_j}$. By the symmetrization step of the protocol, $\rho_{j,0}=\rho_{j,1}$ for $j=1,\ldots,r-1$. Therefore, with the triangle inequality of the trace norm, we have
    \begin{equation}\nonumber
        D( \rho_0, \rho_{r-1,1} ) \leq 3 \sum_{j=1}^{r-1} \sqrt{p_j}.
    \end{equation}
    From the assumption of the soundness, $\mathrm{tr}(M_{y,0}\rho_0) \geq \frac{2}{3}$. Then, by the linearity of the trace and the property of the trace norm, an inequality follows as
    \begin{equation}\nonumber
        p_r = \mathrm{tr}(M_{y,0}\rho_{r-1,1}) = \mathrm{tr}(M_{y,0} \rho_0) -  \mathrm{tr}(M_{y,0} (\rho_0 - \rho_{r-1,1})) \geq \frac{2}{3} - \|\rho_0 - \rho_{r-1,1} \|_1 \geq \frac{2}{3} - 3 \sum_{j=1}^{r-1} \sqrt{p_j}.
    \end{equation}
    Since $0 \leq p_j \leq 1$, we have
    \begin{equation}\nonumber
        3 \sum_{j=1}^r \sqrt{p_j} = 3 \sqrt{p_r} + 3\sum_{j=1}^{r-1} \sqrt{p_j} \geq \sqrt{p_r} + 3\sum_{j=1}^{r-1} \sqrt{p_j} \geq p_r + 3\sum_{j=1}^{r-1} \sqrt{p_j} \geq \frac{2}{3}.
    \end{equation}
    From the Cauchy-Schwarz inequality, we get
    \[
        \sqrt{r}\sqrt{\sum_{j=1}^r p_j} \geq \sum_{j=1}^r \sqrt{p_j}.
    \]
    We thus conclude
    \[
        \sum_{j=1}^r p_j \geq \left(\frac{1}{\sqrt{r}}\sum_{j=1}^r \sqrt{p_j} \right)^2 \geq \left( \frac{2}{3 \cdot 3 \sqrt{r}} \right)^2 = \frac{4}{81r},
    \]
    as claimed.
\end{proof}
By Lemma \ref{lemma:prob}, we have
\begin{eqnarray*}
    \mathrm{Pr}[\neg{E_1} \lor \neg{E_2} \lor \cdot \cdot \cdot \lor \neg{E_r}] &\geq& \frac{1}{r} \sum_{j=1}^r \mathrm{Pr}[\neg{E_j}]. \\
    &\geq& \frac{4}{81r^2},
\end{eqnarray*}
which implies that the protocol $\mathcal{P}_\pi$ has soundness $1-\frac{4}{81r^2}$.

\subsubsection*{Full protocol}

Let us consider a $k$-times repetition of the protocol $\mathcal{P}_\pi$ to reduce the soundness error which is a standard technique for $\QMA$ as in \cite{AN02,KSV02}. The protocol $\mathcal{P}_\pi[k]$ described in Algorithm \ref{algorithm:parallel_repetition} has soundness $(1-\frac{4}{81r^2})^k$. Let us set $k = \lceil 2 \frac{81r^2}{4} \rceil$ and then the protocol has soundness $(\frac{1}{e})^2 < \frac{1}{3}$. The proof size is $O(r^2 \log n)$ qubits for each node and the communication amount between nodes is $O(r^2 \log n)$ respectively.

\begin{algorithm}[H]
\caption{\, Protocol $\mathcal{P}_\pi[k]$ }\label{algorithm:parallel_repetition}
\begin{algorithmic}[1]
\State The prover sends $2k$ quantum registers $R_{j,0,i}, R_{j,1,i}$ for $i \in \{1,\ldots,k\}$, which are $c \log n$ qubits respectively, as proofs to each of the intermediate nodes $v_j$ for $j \in \{1,\ldots,r-1\}$. 
\State The left-end node $v_0$ prepares $k$ states $(\ket{h_x}\bra{h_x})^{\otimes k}$ in the registers $R_{0,1,i}$ for  $i \in \{1,\ldots,k\}$ by itself. 
Then $v_0$ sends their registers to $v_1$. 
\State Each intermediate node $v_j$ swaps the states between $R_{j,0,i}$ and $R_{j,1,i}$ with probability $\frac{1}{2}$, i.e., symmetrizes the states on $R_{j,0,i}$ and $R_{j,1,i}$.
\State Each intermediate node $v_j$ sends $R_{j,1,i}$ to the right neighbor $v_{j+1}$ for all $i \in \{1,\ldots,k\}$.
\State Each intermediate node $v_j$ receives $k$ quantum registers $R_{j-1,1,i}$ from its left neighbor $v_{j-1}$. Then $v_j$ performs the SWAP test on the registers $(R_{{j-1},1,i}, R_{j,0,i})$ for each $i \in \{1,\ldots,k\}$. The node $v_j$ rejects if at least one of the performed SWAP tests rejects, and accepts otherwise.
\State The right-end node $v_r$ receives $k$ registers $R_{r-1,1,i}$ from its left neighbor. Then, $v_r$ performs the POVM measurement $\{M_{y,0},M_{y,1}\}$ corresponding to $\pi$ applied to the states $R_{r-1,1,i}$. The node $v_r$ rejects if at least one of the performed POVM measurements rejects, and accepts otherwise.
\end{algorithmic}
\end{algorithm}

\subsection{Protocol on general graphs}

Let $G=(V,E)$ be a network of radius $r$ with terminals $u_1,\ldots,u_t$. 
Let us assume, without loss of generality, that $u_1$ is the most central node among them, i.e., it satisfies $\max_{i=1,\ldots,t}\mathsf{dist}_G(u_1,u_i) = \min_{j=1,\ldots,t}\max_{i=1,\ldots,t}\mathsf{dist}_G(u_j,u_i)$. Let us construct a tree $T$ rooted at $u_1$, with the other terminals as leaves, maximum degree $t$, and depth at most $r+1$. To do this, we start with the breadth-first search from $u_i$ and find a tree $T'$. Then, we truncate $T'$ at each terminal $u_i$ that does not have any terminal as successors, which limits the depth of the tree to $r$ and the maximum degree to $t$. For every terminal $u_i$ that is not a leaf, replace $u_i$ and connect $u_i$ to $u_i'$ as a leaf, where $u_i$ keeps the input $x_i$. By this construction, we ensure all the terminals have degree $1$ and the depth can be increased by at most 1. See Figure 1 in \cite{FGNP21} for an illustration of the construction. Any protocol of $u_i$ and $u_i'$ over $T$ is simulated on the node $u_i$ over $T'$, which does not affect the soundness and completeness of $\dMA$ and $\dQMA$ protocols from their definitions.

It is also known that there exists a deterministic $\dMA$ protocol that checks if a tree $T$ satisfies the condition.

\begin{lemma}[\cite{Pel00,KKP10}]\label{lemma:KKP10}
	For any network $G=(V,E)$ with nodes IDs taken in a range polynomial in~$|V|$, there is a deterministic $\dMA$ protocol (i.e., with completeness $1$ and soundness $0$) for the tree $T$ using a proof of $O(\log |V|)$ bits for each node. 
\end{lemma}

Based on the tree construction and the deterministic $\dMA$ protocol above, we can focus on a protocol over the tree $T$ since if any malicious prover tells a fake tree construction over nodes, at least one node can detect it with certainty. 

Now we present a $\dQMA$ protocol for the equality function $\mathsf{EQ}^{t}_{n}$, 
which is a function from $(\{0,1\}^n)^t$ to $\{0,1\}$ defined as 
$\mathsf{EQ}^{t}_{n}(x_1,\ldots,x_t)=1$ if $x_1=\cdots=x_t$ and $0$ otherwise. 

\begin{theorem}\label{theorem:eq_tree}
    There exists a $\dQMA$ protocol for $\mathsf{EQ}^{t}_{n}$ 
    on a network $G$ of radius $r$ with perfect completeness and soundness $\frac{1}{3}$, 
    using local proof and message of size $O(r^2 \log n)$.
\end{theorem}

\begin{proof} Our protocol assuming a spanning tree $T$ rooted at $u_1$ guaranteed by Lemma~\ref{lemma:KKP10} is described as Algorithm~\ref{algorithm:perm_equality}.

\begin{algorithm}[H]
\caption{\, Protocol $\mathcal{P}(\mathsf{EQ}_n^t)$ on a spanning tree $T$}\label{algorithm:perm_equality}
\begin{algorithmic}[1]
\State The prover sends two $c \log n$-qubit states in registers $R_{v,0}$ and $R_{v,1}$ to each of the nodes $v$ which has no input. Then, $v$ symmetrizes the two $c \log n$-qubit states on $R_{v,0}$ and $R_{v,1}$.
\State For every $i \in \{1,\ldots,t\}$, the node $u_i$ prepares the $c\log n$-qubit state $\ket{h_{x_i}}$ in register $R_{u_i,1}$.
\State Every non-root node $v$ of the tree sends its $c\log n$-qubit state in $R_{v,1}$ to its parent in $T$.
\State Every non-terminal node $v$ receives some $c \log n$-qubit states from the children. Then, it performs the permutation test on states that consist of the $c \log n$-qubit state received from the prover and the $c \log n$-qubit states received from the children. Then, it accepts or rejects accordingly.
\State The root node $u_1$ receives some $c \log n$-qubit states form its children. Then $u_1$ performs the permutation test on the state that consist of $\ket{h_{x_1}}$ and the states from the children. Then, accept or reject accordingly.
\end{algorithmic}
\end{algorithm}

The perfect completeness follows from Lemma \ref{lemma:perm_ac} since fingerprints $|h_{x_i}\rangle$ for $i=1,\ldots,t$ are the same. For the soundness, let us assume $\mathsf{EQ}_n^t(x_1,\ldots,x_t)=0$, i.e., there is a leaf $u_i$ whose input $x_i$ is not equal to $x_1$. Then, a similar analysis holds as in Section \ref{subsection:path} for the path connecting $u_1$ and $u_i$. This is because the analysis of Lemma \ref{lemma:sum_rej} holds even if some of the nodes on the path conduct the permutation test instead of the SWAP test due to Lemma  \ref{lemma:perm_close}. Therefore, $\mathcal{P}(\mathsf{EQ}_n^t)$ has soundness $1-O(\frac{1}{r^2})$. By the parallel $O(r^2)$ repetitions of $\mathcal{P}(\mathsf{EQ}_n^t)$ similar to the protocol $\mathcal{P}_\pi$, the soundness error can be reduced to $\frac{1}{3}$ and thus the proof of Theorem \ref{theorem:eq_tree} is completed.
\end{proof}

Finally, we can combine the technique to replace quantum communication with classical communication by \cite{GMN23a} with our result. 
If the communication at the verification stage (i.e., the communication among the nodes) of a $\dQMA$ protocol is classical, \cite{GMN23a} named it an $\mathsf{LOCC}$ (Local Operation and Classical Communication) $\dQMA$ protocol. In \cite{GMN23a}, the following result was obtained. 

\begin{lemma}[Theorem 5 in \cite{GMN23a}]\label{lemma:LOCC_dQMA}
    For any constant $p_c$ and $p_s$ such that $0 \leq p_s < p_c \leq 1$, let $\mathcal{P}$ be a $\dQMA$ protocol for some problem on a network $G$ with completeness $p_c$, soundness $p_s$, local proof size $s^\mathcal{P}_c$ and local message size $s^\mathcal{P}_m$. For any small enough constant $\gamma > 0$, there exists an $\mathsf{LOCC}$ $\dQMA$ protocol $\mathcal{P'}$ for the same problem on $G$ with completeness $p_c$, soundness $p_s + \gamma$, local proof size $s^\mathcal{P}_c + O(d_{max} s^\mathcal{P}_m s^\mathcal{P}_{tm})$, and local message size $O(s^\mathcal{P}_m s^\mathcal{P}_{tm})$, where $d_{max}$ is the maximum degree of $G$, and $s^\mathcal{P}_{tm}$ is the total number of qubits sent in the verification stage of $\mathcal{P}$.
\end{lemma}

Theorem~\ref{theorem:eq_tree} and Lemma~\ref{lemma:LOCC_dQMA} lead to the following corollary, which shows a more efficient $\mathsf{LOCC}$ $\dQMA$ protocol for the equality function than Corollary 1 in \cite{GMN23a}.

\begin{corollary}
    For any small constant $\epsilon > 0$, there is an $\mathsf{LOCC}$ $\dQMA$ protocol for $\mathsf{EQ}^{t}_{n}$ on a network $G=(V,E)$ of radius $r$ with completeness $1$, soundness $\epsilon$, local proof size $O(d_{max} |V| r^4 \log^2(n))$ and message size $O(|V| r^4 \log^2(n))$.
\end{corollary}
\section{Robust quantum advantage for $\EQ$ on a path}\label{section:robust_eq}

In this section, we consider the path $v_0,\ldots,v_r$ as a network topology, and $v_0$ and $v_r$ have $n$-bit input strings $x$ and $y$, respectively. We will show that a quantum advantage of distributed verification protocols for the equality problem ($\EQ$) still persists even when the size of the network $r$ is not so small compared with the size of the inputs $n$.
\subsection{Quantum upper bound}\label{subsection:robust_eq_quantum_upper_bound}

In this subsection, we give a $\dQMA$ protocol that is efficient even when the network size is not small.

\begin{theorem}\label{theorem:robust_QA}
    There exists a $\dQMA$ protocol to solve $\EQ$ on the path of length $r$ with total proof size $\sum_i c(v_i) = \Tilde{O}(r n^{\frac{2}{3}})$, perfect completeness and soundness $\frac{1}{3}$.
\end{theorem}

\begin{proof}
Let us denote by $S$ a set of nodes such that the indexes can be divided by $\lceil n^\frac{1}{3} \rceil$, i.e., 
$S = \bigg\{
v_{\lceil n^\frac{1}{3} \rceil},v_{2 \lceil n^\frac{1}{3} \rceil},\ldots, v_{\Bigl\lfloor \frac{r}{\lceil n^\frac{1}{3} \rceil} \Bigr\rfloor\lceil n^\frac{1}{3} \rceil }  \bigg\}
$. 
Let us call nodes of $S$ relay points. Then, the protocol can be described as Algorithm \ref{algorithm:robust_eq}.

The total size of the proof is 
$$O(n^{\frac{2}{3}} \log n) \times (\lceil n^\frac{1}{3} \rceil-1)
\times\bigg(\frac{r}{\lceil n^\frac{1}{3} \rceil}+1\bigg)
+n \times \Bigl\lfloor \frac{r}{\lceil n^\frac{1}{3} \rceil}\Bigl\rfloor
=
\Tilde{O}(r n^{\frac{2}{3}}).
$$ 

To show completeness, let us assume $x=y$. Then, when the proofs for $v_i \in S$ are $\ket{x}$ and the proofs for $v_i \notin S$ are $\ket{h_x}^{\otimes 42 (\lceil n^\frac{1}{3} \rceil)^2}$, all the SWAP tests accept. To show soundness, let us assume $x \neq y$. Then, for any quantum proof, $n$-bit measurement results of at least one adjacent pair of the relay points differ. Then, between the two relay points, at least one node outputs reject from the soundness of the protocol $\mathcal{P}_\pi [42r^2]$ in Algorithm \ref{algorithm:parallel_repetition} with probability $\frac{2}{3}$ as claimed.
\end{proof}

\begin{algorithm}[H]
\caption{\, Protocol for $\EQ$ with ``relay points''}\label{algorithm:robust_eq}
\begin{algorithmic}[1]
\State The prover sends an $n$-qubit state to the relay points $v_i \in S$.
\State The prover sends two $42 (\lceil (n^\frac{1}{3}) \rceil)^2 c \log n$-qubit states to each of the intermediate nodes $v_i \notin S$. Then, the nodes symmetrize the states.
\State On the relay points, the node $v_i \in S$ measures the proof in the computational basis. Based on the $n$-bit measurement results, the nodes create 2 $\times$ $42 (\lceil n^\frac{1}{3} \rceil)^2$ fingerprints (see Section \ref{subsection:preliminary_communication_complexity} for a formal definition of the quantum fingerprints).
\State The left-end node creates $42 (\lceil n^\frac{1}{3} \rceil)^2$ fingerprints $\ket{h_x}$. The right-end node creates $42 (\lceil n^\frac{1}{3} \rceil)^2$ fingerprints $\ket{h_y}$.
\State Each node except the right-end node sends a $42 (\lceil n^\frac{1}{3} \rceil)^2 c \log n$-qubit state to the right neighbor. Then, each node except the left-end node conducts the SWAP test $42 (\lceil n^\frac{1}{3} \rceil)^2$ times on the own fingerprints and the fingerprints from the left neighbor. If even at least one the SWAP test rejects, each node rejects. Otherwise, each node accepts.
\end{algorithmic}
\end{algorithm}

\subsection{Classical lower bound}\label{subsection:classical_lower_bound}
In this subsection, we show that a stronger lower bound of the proof size of $\dMA$ protocols with 1-round verification for $\EQ$.

Let us first show that a linear size proof is required for each local 2 nodes. This is a corollary of Theorem 9 in \cite{FGNP21} but we give a proof for completeness.
\begin{lemma}\label{lemma:classical_lower_bound}
    Let $f(x,y)$ be any Boolean function with a 1-fooling set of size at least $k$. Let $\mathcal{P}$ be a $\dMA$ protocol for $f$ on the path of length $r$, with $\nu$-round of communication among the nodes, shared randomness. Suppose that the proof of size satisfying $\sum_{j=i-\nu+1}^{i+\nu} c(v_j) = \lfloor \frac{1}{2} \log(k-1) \rfloor$ bits for $i \in [\nu,r-\nu-1]$, and $\mathcal{P}$ has completeness $1-p$. Then, $\mathcal{P}$ has soundness error at least $1-2p$.
\end{lemma}

\begin{proof}
For conciseness, we show only the case that $\mathcal{P}$ is a 1-round communication protocol 
(we can easily modify the following proof to the $\nu$-round case). 
Since $f$ has a large $1$-fooling set and the proof size is small, there exist two distinct pairs of fooling inputs that have the same assignment of proofs on $v_i$ and $v_{i+1}$. 
    Let us fix such two inputs pairs $(x,y)$ and $(x',y')$ such that $f(x,y)=f(x',y')=1$ and $f(x,y')=0$ with corresponding assignment of proofs $w$ and $w'$ such that $w(v_{i})=w'(v_{i})$ and $w(v_{i+1})=w'(v_{i+1})$, 
    where $w(v_j)$ is the $v_j$'s part of $w$.
    
    We denote by $\mathsf{out}_i(x,y,w)$ the output of $v_i$ when the inputs are $x$ and $y$ and the proof assignment is $w$. Since $\mathcal{P}$ has completeness $1-p$, we have
    \[
        \Pr_s \big [\bigwedge_{j:j \leq i}\mathsf{out}_j(x,y,w)=1 \wedge \bigwedge_{j: j \geq i+1} \mathsf{out}_j(x,y,w)=1 \big] \geq 1-p,
    \]
    where $s$ denotes the random string taken in ${\cal P}$.  
    The same holds for $(x',y',w')$. Hence,
    \[
        \Pr_s \big [\bigwedge_{j:j \leq i}\mathsf{out}_j(x,y,w)=1 \big] \geq 1-p,
    \]
    \[
        \Pr_s \big [\bigwedge_{j:j \geq i+1}\mathsf{out}_j(x',y',w')=1 \big] \geq 1-p.
    \]
    Let $w''$ be the proof assignment defined by $w''(v_j)=w(v_j)$ for $j \in [0,i-1]$, $w''(v_j)=w(v_j)=w'(v_j)$ for $j=i,i+1$ and $w''(v_j) = w'(v_j)$ for $j \in [i+2,r]$. Consider the input assignment $(x,y')$ combined with the proof assignment $w''$.
    Then, the nodes $v_j$ for $j \leq i$ receive the same partial inputs and proof when the total inputs and proof are $(x,y',w'')$ and $(x,y,w)$ and the nodes $v_j$ for $j \geq i+1$ receive the same partial inputs and proof when the total inputs and proof are $(x,y',w'')$ and $(x',y',w')$. Therefore, by a union bound, we have
    \begin{eqnarray*}
        \lefteqn{ \Pr_s \big [\bigwedge_{j:j \leq i}\mathsf{out}_j(x,y',w'')=1 \wedge \bigwedge_{j:j \geq i+1} \mathsf{out}_j(x,y',w'')=1 \big]
        } \\
        &\geq& 
        1 -\Pr_s \big [\lnot \bigwedge_{j:j \leq i}\mathsf{out}_j(x,y',w'')=1 \big] - \Pr_s \big [\lnot \bigwedge_{j:j \geq i+1}\mathsf{out}_j(x',y',w'')=1 \big] \\
        &=&
        1 - \Pr_s \big [\lnot \bigwedge_{j:j \leq i}\mathsf{out}_j(x,y,w)=1 \big] - \Pr_s \big [\lnot \bigwedge_{j:j \geq i+1}\mathsf{out}_j(x',y',w')=1 \big] \\
        &\geq& 1-2p,
    \end{eqnarray*}
    which implies the soundness error is at least $1-2p$.
\end{proof}

\begin{proposition}\label{proposition:classical_lowerbound}
    Let $f(x,y)$ be any Boolean function with a 1-fooling set of size at least $k$. Let $\mathcal{P}$ be a $\dMA$ protocol for $f$ on the path of length $r$, with $\nu$-round of communication among the nodes, shared randomness, total proof size $\sum_{j={0}}^{r} c(v_j)\leq \lfloor \frac{r-1}{2\nu} \rfloor \lfloor \frac{1}{2} \log(k-1) \rfloor$, and completeness $1-p$. Then, $\mathcal{P}$ has soundness error at least $1-2p$.
\end{proposition}

\begin{proof}
By the pigeonhole principle, there exists $i \in [\nu,r-\nu-1]$ such that $\sum_{j={i}}^{i+1} c(v_j) \leq \lfloor \frac{1}{2} \log(k-1) \rfloor$. Then, Lemma \ref{lemma:classical_lower_bound}, the protocol $\mathcal{P}$ has soundness error at least $1-2p$.
\end{proof}

Since $\EQ$ has a 1-fooling set of size $2^n$, the corollary below directly follows from Proposition $\ref{proposition:classical_lowerbound}$.

\begin{corollary}\label{corollary:classical_lowerbound_for_EQ}
Let $\mathcal{P}$ be any $\dMA$ protocol for $\EQ$ with $\nu$-round of communication between the nodes on the path of length $r$ with total proof size $\sum_{j={0}}^{r} c(v_j) \leq \lfloor \frac{r-1}{2\nu} \rfloor \lfloor \frac{1}{2} (n-1) \rfloor$, and completeness $1-p$. Then, $\mathcal{P}$ has soundness error at least $1-2p$.
\end{corollary}

Corollary \ref{corollary:classical_lowerbound_for_EQ} implies that any $\dMA$ protocol with constant-round, sufficiently high completeness and low soundness error has to receive $\Omega(rn)$ bits as proofs in total.
\section{Protocol for comparing the values of inputs}\label{section:gt}

In this section, we give $\dQMA$ protocols to compare the values of inputs regarded as integers.

\subsection{Protocol for the greater-than problem}\label{subsection:gt}

In this subsection, we construct an efficient $\dQMA$ protocol for the greater-than function ($\GT$). 

The function $\mathsf{GT}:\{0,1,\ldots,2^n-1\}\times\{0,1,\ldots,2^n-1\}\rightarrow\{0,1\}$ is defined as $\GT(x,y)=1$ if and only if $x > y$. 
We identify
\begin{eqnarray*}
    x &=&   { x_0 \times 2^{n-1} + x_1 \times 2^{n-2} + \cdots + x_{n-2} \times 2^1 + x_{n-1} \times 2^0 },\\
    y &=&  { y_0 \times 2^{n-1} + y_{1} \times 2^{n-2} + \cdots + y_{n-2} \times 2^1 + y_{n-1} \times 2^0 }
\end{eqnarray*}
by the $n$-bit strings 
$x={x_0 x_1\cdots x_{n-2} x_{n-1}}$ 
and $y={y_0 y_1\cdots y_{n-2} y_{n-1}}$.

We first observe $\GT(x,y)=1$ if and only if there exists an index $i \in [0,n-1]$ such that $x_i=1$, $y_i=0$ and $x[i]=y[i]$, where $x[i]:=x_0\cdots x_{i-1}$ and $y[i]:=y_0\cdots y_{i-1}$. 
Then, we construct a $\dQMA$ protocol for $\GT$ using the protocol for $\EQ$ as a subroutine.

\begin{theorem}\label{theorem:greater-than}
    There exists a $\dQMA$ protocol for $\GT$ on the path of length $r$ with perfect completeness and soundness $\frac{1}{3}$, using local proof and message of size $O(r^2\log n)$.
\end{theorem}

\begin{proof}
The protocol before the parallel repetition can be described in Algorithm \ref{algorithm:greater-than}. 

\begin{algorithm}
\caption{\, Protocol for $\GT$ on an input pair $(x,y)$ in a path $v_0,\ldots,v_r$}
\label{algorithm:greater-than}
\begin{algorithmic}[1]
\State The prover sends two $O(\log n)$-qubit registers $R_{j,0}, R_{j,1}$ called fingerprint registers to each of the intermediate nodes $v_j$ for $j \in \{1,\ldots,r-1\}$. The prover also sends a $\lceil \log n \rceil$-qubit register, called an index register, to each of all the nodes.  
\State The node $v_0$ measures the index register in the computational basis and let us denote by $i_0 \in \{0,1\}^{\lceil\log n\rceil}$ the measurement result. If $x_{i_0}=0$, $v_0$ rejects.  
Then, $v_0$ prepares the state 
$\rho_0 = \ket{h_{x[i_0]}}\bra{h_{x[i_0]}}$ in register $R_0$ as the fingerprint of the binary string $x[i_0]:=x_{0}\cdots x_{i_0-1}$. 
If $i_0=0$, it prepares $\ket{\perp}$. 
\State Each intermediate node $v_j$ measures the index register in the computational basis. It also swaps the states between $R_{j,0}$ and $R_{j,1}$ with probability $\frac{1}{2}$, i.e., symmetrizes the states on $R_{j,0}$ and $R_{j,1}$. \label{step:gt;path_symmetrize}
\State The node $v_r$ measures the index register in the computational basis and let us denote by $i_r \in \{0,1\}^{\lceil\log n\rceil}$ the measurement result. If $y_{i_r}=1$, $v_r$ rejects. 
Then, $v_r$ prepares the state $\rho_r = \ket{h_{y[i_r]}}\bra{h_{y[i_r]}}$ in register $R_r$ as the fingerprint of the binary string $y[i_r]:=y_{0} \cdots y_{i_r-1}$. If $i_r=0$, it prepares $\ket{\perp}$. 
\State The node $v_0$ sends $R_0$ and a register $R'_{0}$ that encodes the measurement result of the index register to the right neighbor $v_1$. 
Each intermediate node $v_j$ sends $R_{j,1}$ and a register $R'_{j}$ that encodes the measurement result of the index register to the right neighbor $v_{j+1}$. 
\State Each intermediate node $v_j$ receives $R_{{j-1},1}$ from its left neighbor $v_{j-1}$. The node $v_j$ also receives $R'_{j-1}$ from $v_{j-1}$ and measures them in the computational basis to check if the measurement result is the same as the own index register or not. If they are different, $v_j$ rejects. Otherwise, $v_j$ performs the SWAP test on the registers $(R_{{j-1},1}, R_{j,0})$ and accepts or rejects accordingly. \label{step:gt;path_swap}
\State The node $v_r$ receives $R_{r-1,1}$ and $R'_{r-1}$ from its left neighbor $v_{r-1}$. Then, $v_r$ measures $R'_{r-1}$ in the computational basis and checks if the measurement result is the same as the own index register. If they are different, $v_r$ rejects. Otherwise, $v_r$ performs the SWAP test on $(R_{r-1,1},R_r)$, and accepts or rejects accordingly. 
\end{algorithmic}
\end{algorithm}


\subsubsection*{Completeness}
Let us assume $\GT(x,y)=1$, i.e., $x>y$. Then, there exists an index $i$ such that $x_i=1$, $y_i=0$, and 
$x[i]=y[i]$. To achieve perfect completeness, 
the honest prover can send the index $i$ in the index register and $\ket{h_{x[i]}}=\ket{h_{y[i]}}$ in the fingerprint register to all the nodes. 
If $i=0$, the prover sends the index $0$ in the index register and $\ket{\perp}$ in the fingerprint register. Then, all the nodes accept since $x_i=1$, $y_i=0$, and all the index comparisons and the SWAP tests are accepted.

\subsubsection*{Soundness}
From the index comparisons in the protocol, 
the prover must send the same index in all the index registers to maximize the acceptance probability. 
Thus we assume the prover sends the same index $i$ in all the index registers. 

Let us assume $\GT(x,y)=0$, i.e., $x \leq y$. 
If $x_i=0$ or $y_i=1$, $v_0$ or $v_r$ rejects. Thus the prover must choose $i$ such that $x_i=1$ and $y_i=0$. Then $x[i] \neq y[i]$ as otherwise $x>y$, which contradicts $\GT(x,y)=0$. 
(Note that when $i=0$, $x_0=0$ or $y_0=1$ holds from $x\leq y$ and thus $v_0$ or $v_r$ rejects. Hence we can assume $i \geq 1$.) Then, by the soundness analysis of the $\dQMA$ protocol for $\EQ$, at least one node rejects with probability $O(\frac{1}{r^2})$. 

By the parallel repetition of Algorithm \ref{algorithm:greater-than} with $O(r^2)$ times, the protocol has a sufficiently low constant soundness error.
This completes the proof.
\end{proof}

Since the size of the 1-fooling set of $\GT$ is $2^n$, a lower bound of $\dMA$ protocols can be shown from Proposition~$\ref{proposition:classical_lowerbound}$.

\begin{corollary}
Let $\mathcal{P}$ be any $\nu$-round $\dMA$ protocol for $\GT$ on the path of length $r$ with total proof size $\sum_{j={0}}^{r} c(v_j) \leq \lfloor \frac{r-1}{2\nu} \rfloor \lfloor \frac{1}{2} (n-1) \rfloor$ and completeness $1-p$. Then, $\mathcal{P}$ has soundness error at least $1-2p$.
\end{corollary}

We can define three functions from $\{0,1,\ldots,2^n-1\}\times\{0,1,\ldots,2^n-1\}$ to $\{0,1\}$, $\GT_{<}$, $\GT_{\geq}$, and $\GT_{\leq}$ as follows: 
$\GT_{<}(x,y)=1$ iff $x<y$, $\GT_{\geq}(x,y)=1$ iff $x\geq y$, and $\GT_{\leq}(x,y)=1$ iff $x\leq y$.\footnote{$\GT$ can be regarded as $\GT_{>}$ in this notation.} By modifying our protocol for $\GT$, we also obtain $\dQMA$ protocols for these functions.

\begin{corollary}\label{cor:GT-variants}
    There are $\dQMA$ protocols for $\GT_{<}$, $\GT_{\geq}$, and $\GT_{\leq}$ on the path of length $r$ with perfect completeness and soundness $\frac{1}{3}$ and using local proof and message of size $O(r^2\log n)$.
\end{corollary}

\subsection{Application for ranking verification}\label{subsection:rv}

In this subsection, we apply the $\dQMA$ protocol for $\GT$ for verifying the ranking of a terminal in a network.

Let us define the ranking verification problem, 
which asks whether the input $x_i$ of the $i$-th terminal is the $j$-the largest over all the inputs over the network.

\begin{definition}[ranking verification]\label{definition:ranking}
    For $i,j \in [1,t]$, $\mathsf{RV}^{i,j}_{t,n}(x_1,\ldots,x_t)=1$ if and only if 
    $$
    \sum_{k\in [1,t] \setminus \{i\} } \GT_{\geq}(x_i,x_k)=t-j+1.
    $$ 
\end{definition}

By running the $\dQMA$ protocol for $\GT_{\geq}$ (and $\GT_{<}$) in parallel on a spanning tree rooted at $u_i$, we obtain a $\dQMA$ protocol for $\mathsf{RV}$.

\begin{theorem}\label{theorem:ranking_verification}
    For $i,j \in [1,t]$, there exists a $\dQMA$ protocol 
    for $\mathsf{RV}^{i,j}_{t,n}$ with $t$ terminals and radius $r$, with perfect completeness and soundness $\frac{1}{3}$, using local proof and message size $O(tr^2 \log n)$.
\end{theorem}

\begin{proof}
Let $u_1,\ldots,u_t$ be the $t$ terminals where $x_k$ is owned by $u_k$. The protocol is described as Algorithm \ref{algorithm:ranking_verification}. 

    \begin{algorithm}[H]
        \caption{\, Protocol for $\mathsf{RV}^{i,j}_{t,n}$}\label{algorithm:ranking_verification}
        \begin{algorithmic}[1]
            \State An honest prover tells a spanning tree $T$ whose root is $u_i$ and leaves are the other terminals.
            \State For every leaf terminal $u_k$ and every node on the path between $u_i$ and $u_k$ in $T$, a $1$-qubit register called a direction register is sent from the prover, where $0$ and $1$ in the direction register represent $``\geq"$ (which means $x_i\geq x_k$) and $``<"$ (which means $x_i< x_k$), respectively.
            Moreover, the prover sends a proof $\rho$ to the nodes on the path according to the protocol for $\GT_{\geq}$ (when $x_i\geq x_k$) or $\GT_{<}$ (when $x_i<x_k$).  
            \State For the node $u_i$ and each of the other terminals $u_k$, the following steps are done: (i) Check whether all the contents of the direction registers on the path between $u_i$ and $u_k$ are the same or not using $1$-bit information obtained by measuring each direction register in the computational basis. (ii) If all the contents are $``\geq"$ (resp.~$``<"$), the nodes on the path conduct the protocol for $\GT_{\geq}$ (resp.~$\GT_{<}$) using the proof $\rho$ from the prover. 
            \State The root node $v_i$ counts the number of $``\geq"$ in the $t-1$ direction registers from the prover, and rejects if $\sum_{k\in [1,t] \setminus \{i\} } \GT_{\geq}(x_i,x_k) \neq t-j+1$. Otherwise, $v_i$ accepts.
        \end{algorithmic}
    \end{algorithm}

The local proof and message sizes are $O(tr^2 \log n)$ as every node receives at most $t-1$ fingerprint registers whose size is guaranteed by Corollary~\ref{cor:GT-variants}. 

In the following analysis, we can assume that all nodes on the path between $u_1$ and any leaf $u_k$ receive the same direction ($\geq$ or $<$) in the direction registers as otherwise the prover is rejected with probability $1$.

The completeness holds because the honest prover can 
send the true direction for each path, namely, $\geq$ (resp.~$<$) is chosen when $x_i\geq x_k$ (resp.~$x_i<x_k$). Then all the protocols for $\GT_{\geq}$ or $\GT_{<}$ accept and the root node $u_i$ also accepts at the final step since the number of $``\geq"$ is exactly $t-j+1$. 

To show the soundness, let us assume that $\mathsf{RV}^{i,j}_{t,n}(x_1,\ldots,x_t)=0$. If the prover sends the true direction and follows the corresponding protocol for $\GT_{\geq}$ (resp.~$\GT_{<}$) 
according to $x_i\geq x_k$ (resp.~$x_i<x_k$) 
for every leaf $x_k$, the root node $u_i$ rejects at the final step since $\sum_{k\in [1,t] \setminus \{i\} } \GT_{\geq}(x_i,x_k) \neq t-j+1$. 
Thus, the prover must send a false direction and cheat the protocol for $\GT_{\geq}$ or $\GT_{<}$ on some path. 
However, from the soundness of the protocol for $\GT_{\geq}$ or $\GT_{<}$ (by Corollary~\ref{cor:GT-variants}), the probability that at least one node on the path rejects is at least $\frac{2}{3}$. 
\end{proof}

\section{Protocol for the Hamming distance and beyond on general graphs}\label{section:hamming_distance}

In this section, we derive $\dQMA$ protocols for the Hamming distance and more extended functions on general graphs.

\subsection{Protocol for the Hamming distance}

$\mathsf{HAM}^{\leq d}_n(x,y)=1$ if and only if the Hamming distance between $n$-bit strings $x$ and $y$ is at most $d$.
The SMP (and hence one-way) quantum communication complexity of $\mathsf{HAM}^{\leq d}_n(x,y)$ 
is $O(d\log n)$ \cite{LZ13}, improving the previous works \cite{Yao03,GKW04}. Let $c'$ be an enough large constant independent with $n$, $r$ and $d$, and let $\pi'$ be a quantum one-way communication protocol for the Hamming distance transmitting $c' d \log n$ qubits from \cite{LZ13}, such that, for all input pairs $(x,y)$, if $\mathsf{HAM}^{\leq d}_n(x,y)=1$ then $\pi'$ outputs 1 with probability at least $\frac{2}{3}$, and if $\mathsf{HAM}^{\leq d}_n(x,y) = 0$, then $\pi'$ output $0$ with probability at least $\frac{2}{3}$. Let $|\psi(x)\rangle$ be the $c' d \log n$-qubit (pure) state sent from Alice to Bob in $\pi'$ when $x$ is an input for Alice.

As a previous work, 
there is a $\dQMA$ protocol 
for the Hamming distance problem on a path network.

\begin{fact}[Corollary 3 in \cite{FGNP21}]
    For any $c>0$ and $d \in \mathbb{Z}$, there exists a 
    $\dQMA$ protocol for $\mathsf{HAM}^{\leq d}_n$ on the path of length $r$ with completeness $1-\frac{1}{n^c}$, soundness $\frac{1}{3}$, and using local proof and message of size $O(r^2 d (\log n) \log(n+r))$.
\end{fact}

We generalize the $\dQMA$ protocol for the Hamming distance between multiple inputs over apart nodes on a network. As in the case of the equality function, let $r$ be the radius and $t$ be the number of the terminals.
The function of Hamming distance for $t$ terminals $u_1,\ldots,u_t$ where $u_j$ has an $n$-bit string $x_j$ can be defined as follows; $\mathsf{HAM}^{\leq d}_{t,n}(x_1,\ldots,x_t)=1$ if and only if 
the Hamming distance between any two $n$-bit strings $x_i$ and $x_j$ is at most $d$. Then, we show the following theorem.

\begin{theorem}\label{theorem:hamming_distance}
    For any $c>0$ and $d \in \mathbb{Z}$, there exists a $\dQMA$ protocol for $\mathsf{HAM}^{\leq d}_{t,n}$ on a network of radius $r$ with completeness $1-\frac{1}{n^c}$ and soundness $\frac{1}{3}$, using local proof and message of size $O(t^2 r^2 d (\log n) \log (n+t+r))$.
\end{theorem}

\begin{proof}
Let us first consider a two-sided error one-way protocol $\pi''$ that repeats the one-way communication protocol $\pi'$ for $O(\log (n+t+r))$ times and takes a majority of the outcomes to reduce the error probability. The protocol $\pi''$  on input $(x,y)$ accepts with probability at least $1-\frac{1}{42 n^c t^2 r^2}$ when $\mathsf{HAM}^{\leq d}_n(x,y)=1$ and accepts with probability at most $\frac{1}{3}$ when $\mathsf{HAM}^{\leq d}_n(x,y)=0$.\footnote{Actually, it is at most $\frac{1}{42 n^c t^2 r^2}$, which is smaller than $\frac{1}{3}$.} Note that $|\psi''(x)\rangle:=|\psi(x)\rangle^{\otimes O(\log(n+t+r))}$ is the state from Alice on input $x$ to Bob in $\pi''$. 

As our $\dQMA$ protocol for $\EQ$, we assume that the network can know the construction of the spanning tree whose root and leaves are terminals from the prover. In the $\EQ$ protocol, we considered a protocol where messages are sent from the leaves to the root. In contrast, let us consider a protocol where messages are sent from the root to the leaves to show the completeness of the protocol. We also consider running the protocols in parallel for all the $t$ spanning trees whose roots are the $t$ terminals to show the soundness of the protocol. The total verification algorithm can be described in Algorithm \ref{algorithm:Hammingdistance}.

\begin{algorithm}
\caption{\, Protocol for the Hamming distance on general graphs}\label{algorithm:Hammingdistance}
\begin{algorithmic}[1]
\State The (honest) prover sends $t$ spanning trees $T_1,\ldots,T_t$ to the nodes: 
the root of the $j$th one is the $j$th terminal $u_j$, and the leaves are the other terminals.
\For{$j=1,\ldots,t$}
\State The honest prover sends $(\delta+1)$ quantum registers ${\sf R}_{j,v,1},\ldots,{\sf R}_{j,v,\delta+1}$ 
to a node $v$ which is neither a root nor a leaf and whose number of its children is $\delta$. The contents of the registers are assumed to be the fingerprint $|\psi''(x_r)\rangle$ 
of the root $u_r$. Then, $v$ permutes $(\delta+1)$ registers by a permutation on $S_{\delta+1}$ chosen uniformly at random. Then $v$ keeps ${\sf R}_{j,v,\delta+1}$ (renamed by the permutation), and sends ${\sf R}_{j,v,\mu}$ to the $\mu$th child of $v$.
\State The root node $u_r$ with input $x_r$ sends the fingerprint $|\psi''(x_r)\rangle$ to each of the children.
\State Each of non-root nodes, $v$, implements the SWAP test on ${\sf R}_{j,v,\delta+1}$ and the register sent from the parent. Then, $v$ accepts or rejects based on the result of the SWAP test.
\State Each leaf $u_l$ with input $x_l$ does the POVM operation of Bob in the one-way communication protocol $\pi''$ on the register sent from the parent. Then, $u_l$ accepts or rejects based on the result of the POVM operation.
\State To reduce the soundness error, do the parallel repetition of Steps 3 to 6 with $k$ times similarly to Algorithm \ref{algorithm:parallel_repetition}. Each node rejects if at least one of the performed SWAP tests or the operation of Bob in the one-way communication protocol $\pi''$ rejects, and accepts otherwise.
\EndFor
\end{algorithmic}
\end{algorithm}

Let $k=42r^2$. The total size of the quantum registers $R_{j,v,1},\ldots,R_{j,v,\delta}$ is $O(t d\log(n) \log (n+t+r))$ because $\delta$ can be bounded by $t-1$. By the for-loop at Step 2 and the $k$ parallel repetitions at Step 7, the local proof and message sizes are $O(t^2r^2d \log(n) \log(n+t+r))$.

To show the completeness, let us assume that $\mathsf{HAM}^{\leq d}_{t,n}(x_1,\ldots,x_t)=1$. The operations of Bob in the protocol $\pi''$ are done $42r^2t(t-1)$ times in total. Therefore the protocol has completeness $(1-\frac{1}{42t^2r^2n^c})^{42r^2t(t-1)} \geq 1 - \frac{1}{n^c}$.

To show the soundness, let us assume that $\mathsf{HAM}^{\leq d}_{t,n}(x_1,\ldots,x_t)=0$. Then there exist $i$ and $j$ such that $\mathsf{HAM}^{\leq d}_n(x_i, x_j)=0$. 
Over the path on $T_j$ whose extremities are $u_i$ and $u_j$ with input $x_i$ and $x_j$ respectively, the probability that all nodes on the path accept is at most $\left(1-\frac{4}{81r^2}\right)^k<\frac{1}{3}$ by a similar analysis of the $\EQ$ protocol in Section~\ref{section:perm_eq}.
\end{proof}

Since $\EQ$ is a spacial case of $\mathsf{HAM}^{\leq d}_n$ when $d=0$, it can be shown that a similar lower bound of $\dMA$ to Corollary \ref{corollary:classical_lowerbound_for_EQ} holds for $\mathsf{HAM}^{\leq d}_n$.
\begin{corollary}\label{corollary:classical_lowerbound_for_HAM}
Let $\mathcal{P}$ be any $\nu$-round $\dMA$ protocol for $\mathsf{HAM}^{\leq d}_n$ on the path of length $r$ with total proof size $\sum_{j={0}}^{r} c(v_j) \leq \lfloor \frac{r-1}{2\nu} \rfloor \lfloor \frac{1}{2} (n-1) \rfloor$ and completeness $1-p$. Then, $\mathcal{P}$ has soundness error at least $1-2p$.
\end{corollary}

\subsection{Extended results}

In this subsection, we extend Theorem \ref{theorem:hamming_distance} to other problems than the Hamming distance and $\mathsf{LOCC}$ $\dQMA$ protocols.

From a function $f:(\{0,1\}^n)^2 \to \{0,1\}$, let us denote a multi-input function $\forall_t f:(\{0,1\}^n)^t \to \{0,1\}$ where $\forall_t f(x_1,\ldots,x_t) =1$ iff $f(x_i,x_j)=1$ for any $i,j \in [1,t]$.
Similarly to the proof of Theorem~\ref{theorem:hamming_distance}, 
we obtain the following theorem that converts any one-way two-party quantum communication complexity protocol to $\dQMA$ protocols over a network. 

\begin{theorem}\label{proposition:multi-input}
    For a function $f:(\{0,1\}^n)^2 \to \{0,1\}$ such that $\BQP^1(f)=s$, there exists a $\dQMA$ protocol for $\forall_t f$ on a network of radius $r$ with $t$ terminals, completeness $1-\frac{1}{\poly(n)}$ and soundness $\frac{1}{3}$, using local proof and message of size $O(t^2 r^2 s \log(n+t+r))$.
\end{theorem}

We give a number of applications of  Theorem~\ref{proposition:multi-input}. 
First, we apply the techniques of \cite{DM19}. Let us introduce some definitions of $l_1$-graphs. 
Let $V(H)$ denote the set of nodes of a graph $H$. 

\begin{definition}[$l_1$-graph \cite{DL97}]
A graph $H$ is an $l_1$-graph 
if its path metric $\mathsf{dist}_H$ is $l_1$-embeddable, i.e., there is a map $f$ between $V(H)$ and $\mathbb{R}^m$, for some $m$, such that 
$\mathsf{dist}_H(v,w) = \|f(v) - f(w)\|_1$.
\end{definition}

\begin{definition}[$k$-scale embedding \cite{Shp93,BC08}]
    Given two connected and undirected graphs $H$ and $H'$, we say that $H$ is a $k$-scale embedding of $H'$ if there exists a mapping $f : V(H) \rightarrow V(H')$ such that $\mathsf{dist}_{H'}(f(a),f(b)) = k \cdot \mathsf{dist}_H(a,b)$ for all the vertices $a,b \in V(H)$.
\end{definition}

\begin{lemma}[Proposition 8.4 in \cite{BC08}]
    A graph $H$ is an $l_1$-graph if and only if it admits a constant scale embedding into a hypercube.
\end{lemma}

Examples of $l_1$-graphs are Hamming graphs \cite{Che88}, half cubes (the half-square of the hyper cubes) and Johnson graphs \cite{Che17} are $2$-embeddable into a hypercube \cite{BC08}. Using the Johnson-Lindenstrauss lemma \cite{JL84,GKdW06} to reduce the protocol complexity, Driguello and Montanaro showed the following statement as a subroutine of Protocol 2 in \cite{DM19}.

\begin{lemma}[\cite{DM19}]
Let $H = (V,E)$ be an $\ell_1$-graph with $|V|$ vertices, and let $u,v \in V$. There exists a quantum protocol in the $\mathrm{SMP}$ model with private randomness which communicates $O(d^2 \log \log |V|)$ qubits and decide $\mathsf{dist}_H(u,v) \leq d$ or $\mathsf{dist}_H(u,v) \geq d+1$ with arbitrary high constant probability\footnote{Each party knows $H$ in this problem.}.
\end{lemma}

We define a $t$-party version of the above problem.

\begin{definition}
    For an $\ell_1$-graph $H$, $\mathsf{dist}_{t,H}^{\leq d}(v_1,\ldots,v_t) = 1$ 
    if $\mathsf{dist}_H(v_i,v_j) \leq d$ for any distinct $v_i$ and $v_j$ in $H$, and $\mathsf{dist}_{t,H}^{\leq d}(v_1,\ldots,v_t) = 0$ if $\mathsf{dist}_H(v_i,v_j) \geq d+1$ for some distinct $v_i$ and $v_j$ in $H$.
\end{definition}

Then we have the following result from Theorem~\ref{proposition:multi-input}.

\begin{corollary}
\sloppy    
For any $d \in \mathbb{N}$, $\ell_1$-graph $H$, and network $G$ whose radius is $r$ and number of terminals is $t$, 
    there exists a $\dQMA$ protocol for $\mathsf{dist}_{t,H}^{\leq d}(v_1,\ldots,v_t)$ over $G$ with completeness $1-\frac{1}{\poly(\log |V(H)|)}$ and soundness $\frac{1}{3}$, using local proof and message of size $O({t}^2 {r}^2 d^2 \log \log |V(H)| \log (\log(|V(H)|)+t+r))$.
\end{corollary}

Driguello and Montanaro also showed an efficient quantum protocol of the SMP model to distinguish $l_1$-distances between vectors. A special case of the distance is the total variation distance of probabilistic distributions. 

\begin{lemma}[Section I\hspace{-1.2pt}V in \cite{DM19}]
Let $x,y \in [-1,1]^n$ such that each entry of $x$ and $y$ is specified by a $O(n)$-bit string. For any $d>0$, there is a quantum protocol in the SMP model which communicate $O(\frac{\log n}{\epsilon ^2})$ qubits and decide $\| x-y \|_1 \leq d$ or $\| x-y \|_1 \geq d(1+\epsilon)$ for any $\epsilon=\Omega(\frac{1}{\log n})$ with failure probability bounded by an arbitrarily small constant.
\end{lemma}

We can also define a $t$-party version.

\begin{definition}
    For vectors $x_1,\ldots,x_t \in [-1,1]^n$ such that each entry of a vector is specified by a $O(n)$-bit string, $d>0$ and $\epsilon=\Omega(\frac{1}{\log n})$, $\mathsf{dist}_{\mathbb{R}^n}^{\leq d,\epsilon}(x_1,\ldots,x_t)=1$ if $\| x_i-x_j \|_1 \leq d$ for any distinct $i$ and $j$ and $\mathsf{dist}_{\mathbb{R}^n}^{\leq d,\epsilon}(x_1,\ldots,x_t)=0$ if $\| x_i-x_j \|_1 \geq d(1+\epsilon)$ for at least one pair of distinct $i$ and $j$.
\end{definition}

By Theorem~\ref{proposition:multi-input}, the following result is obtained.

\begin{corollary}
    There exists a $\dQMA$ protocol for $\mathsf{dist}_{\mathbb{R}^n}^{\leq d,\epsilon}$ on a network of radius $r$ with $t$ terminals, completeness $1-\frac{1}{\poly(n)}$ and soundness $\frac{1}{3}$, using local proof and message of size $O(t^2r^2\epsilon^{-2}\log n\log(n+r+t))$.
\end{corollary}

A function $F(x,y)$ on $\{0,1\}^n \times \{0,1\}^n$ is an XOR function if $F(x,y) = f(x \oplus y)$ for some function $f$ on $n$-bit strings, where $x \oplus y$ is the bit-wise XOR of $x$ and $y$. An XOR function is symmetric if $f$ is symmetric, i.e., $f(z)$ depends only on the Hamming weight of $z$. The Hamming distance function is indeed an important symmetric XOR function, which can be also defined as follows.
\[
    \mathsf{HAM}_n^{\leq d} (x,y) = \begin{cases}
        1 \ \ \mathrm{if} \ | x \oplus y | \leq d\\
        0 \ \ \mathrm{if} \ | x \oplus y | > d
    \end{cases}
\]

Let us consider more general classes of the XOR function. A linear threshold functions (LTF) $f$ is defined by 
\[
    f(z)=\begin{cases}
        1 \ \ \mathrm{if} \sum_i w_i z_i \leq  \theta \\
        0 \ \ \mathrm{if} \sum_i w_i z_i > \theta
    \end{cases}
\]
where $\{w_i\}$ are the weights and $\theta$ is the threshold. We define 
\[
    W_0 = \max_{z:f(z)=0} \sum _i w_i z_i \hspace{10pt} \mathrm{and} \hspace{10pt} W_1 = \min_{z:f(z)=1} \sum_i w_i z_i,
\]
and define $m_0=\theta-W_0$ and $w_1 = W_1 - \theta$. The margin of $f$ is $m = \max \{m_0,m_1\}$. Note that the function remains the same if $\{ w_i \}$ are fixed and $\theta$ varies in $(W_0,W_1]$. Thus, without loss of generality, we assume that $\theta = \frac{W_0+W_1}{2}$, in which case $m_0=m_1=m$.

\begin{lemma}[Theorem 3 in \cite{LZ13}]
    For any linear threshold function $f$ whose threshold is $\theta$ and margin is $m$ and a function $g$ such that $g(x,y)=f(x \oplus y)$, $\BQP^{||}(g) = O(\frac{\theta}{m}\log n)$.
\end{lemma}

Our multiparty problem and the result induced from Theorem~\ref{proposition:multi-input} 
are given in the following.

\begin{definition}
    For any linear threshold function $f$ whose threshold is $\theta$ and margin is $m$ and $t$ $n$-bit inputs $(x_1,\ldots,x_t)$, $\mathsf{LTF}_n^{\leq \theta,m}(x_1,\ldots,x_t) = 1$ if and only if $f(x_i \oplus x_j)=1$ for any distinct $i$ and $j$.
\end{definition}

\begin{corollary}
    There exists a $\dQMA$ protocol for $\mathsf{LTF}_n^{\leq \theta,m}$ on a network of radius $r$ with $t$ terminals, completeness $1-\frac{1}{\poly(n)}$ and soundness $\frac{1}{3}$, using local proof and message of size $O(t^2r^2 \frac{\theta}{m} \log n\log(n+r+t))$.
\end{corollary}
    
Let us next consider a function $\mathbb{F}_q$-$\mathsf{rank}_{n}^r : \mathbb{F}_q^{n \times n} \times \mathbb{F}_q^{n \times n} \to \{0,1\}$. We define $\mathbb{F}_q$-$\mathsf{rank}_{n}^r (X,Y)=1$ if and only if the matrix $X + Y$ has rank less than $r$, where the rank and the summation $X+Y$ are both over $\mathbb{F}_q$.

\begin{lemma}[Theorem 4 in \cite{LZ13}]
    For $f = \mathbb{F}_q$-$\mathsf{rank}_n^r$, $\BQP^{||}(f) = \min \{ q^{O(r^2)}, O(nr\log q + n \log n) \}$.
\end{lemma}

Our multiparty problem and the result induced from Theorem~\ref{proposition:multi-input} 
are given in the following.

\begin{definition}
    $\mathbb{F}_q$-$\mathsf{rank}_{t,n}^{\leq r}(X_1,\ldots,X_t) = 1$ if and only if $\mathbb{F}_q$-$\mathsf{rank}_n^r (X_i,X_j)=1$ for any distinct $i$ and $j$.
\end{definition}

\begin{corollary}
    \sloppy There exists a $\dQMA$ protocol for $\mathbb{F}_q$-$\mathsf{rank}_{t,n}^{\leq r}(X_1,\ldots,X_t)$ on a network of radius $r$ with $t$ terminals, completeness $1-\frac{1}{poly(n)}$ and soundness $\frac{1}{3}$, using local proof and message of size $O(t^2r^2 \min\{ q^{O(r^2)}, O(nr\log q + n \log n) \} \log(n+r+t))$.
\end{corollary}
\section{Construction of $\dQMA^\mathsf{sep}$ protocols from $\dQMA$ protocols}\label{section:dQMAsep}

In this section, we prove that any function which can be efficiently solved in a $\dQMA$ protocol has an efficient $\dQMA^\mathsf{sep}$ protocol with some overheads. 

We first show that any $\QMA$ one-way communication protocol can be transformed into a $\dQMA$ protocol on a path with some overheads.

\begin{theorem}\label{theorem:transform_QMAcc}
    Suppose that, for a Boolean function $f:\{0,1\}^n \times \{0,1\}^n \to \{0,1\}$, there exists a $\QMA$ one-way communication protocol with a $\gamma$-qubit proof and $\mu$-qubit communications, completeness $\frac{2}{3}$ and soundness $\frac{1}{3}$. Then, there exists a $\dQMA$ protocol for $f$ on a path $v_0,\ldots,v_r$ with completeness $1-\frac{1}{\poly(n)}$ and soundness $\frac{1}{3}$, proof size $c(v_0)=O(r^2 \gamma \log(n+r))$, $c(v_1),c(v_2),\ldots,c(v_{r-1})=O(r^2 (\gamma+\mu) \log(n+r))$, and message size $m(v_i,v_{i+1}) = O(r^2 (\gamma+\mu) \log(n+r))$ for $i \in [0,r-1]$. 
\end{theorem}

\begin{proof}

Let us consider a $O(\log(n+r))$ times repetition of the $\QMA$ one-way communication protocol for $f$ in a standard way as in \cite{AN02,KSV02}. The repeated protocol requires a $O(\gamma \log(n+r))$-qubit proof and a $O(\mu \log(n+r))$-qubit communication from Alice to Bob and has completeness at least $1-\frac{1}{42n^cr^2}$ and soundness at most $\frac{1}{42n^cr^2}$.

Let us describe the above $\QMA$ one-way communication protocol as follows. Note that this formalization holds for any $\QMA$ one-way communication protocol. Merlin produces a quantum state $\rho$ on $\gamma' = O(\gamma \log(n+r))$ qubits, which he sends to Alice. Then, Alice applies some quantum operation on $\rho$ depending on her input $x \in \{0,1\}^n$ and sends Bob a quantum state $\sigma$ on $\mu' = O(\mu \log(n+r))$ qubits. Then, Bob conducts a POVM measurement $\{M_{y,0},M_{y,1}\}$ on the state $\sigma$ depending on his input $y \in \{0,1\}^n$.

To make a quantum message from Alice to Bob in the case of completeness a pure state rather than a mixed state, we consider a variant of the original protocol as follows. Not to confuse readers, let us call two parties Carol and Dave which have an input $x$ and $y$ respectively rather than Alice and Bob. Merlin produces a quantum state $\rho$ on $\gamma'$ qubits, which he sends to Carol. Then, Carol applies some unitary operation $U_x$ on $\rho$ and her $(\gamma'+\mu')$ ancilla qubits and sends Dave a $(2\gamma'+\mu')$-qubit state $\sigma' = U_x (\rho \otimes \ket{0}^{\otimes(\gamma'+\mu')} \bra{0}^{\otimes(\gamma'+\mu')} ) U_x^\dag $. Then, Dave obtains $\sigma$ by tracing out the last $2\gamma'$ qubits of $\sigma'$ and conducts the POVM measurement $\{M_{y,0},M_{y,1}\}$ on the state $\sigma$ depending on his input $y \in \{0,1\}^n$. Let us denote by a POVM measurement $\{M'_{y,0},M'_{y,1}\}$ the whole operations of Dave. This modification can be done from the fact that for any quantum operation (CPTP map) from $n$-qubit to $m$-qubit, there exists an equivalent operation of a unitary matrix on $(2n + m)$-qubit (see, e.g., Lemma 1 in \cite{AKN98}). This modified protocol has the same completeness and soundness to the original protocol.


\begin{algorithm}
\caption{\, $\dQMA$ protocol $\mathcal{P}_{\QMAcc}$ for a function $f$ such that $\QMAcc^1(f)=\gamma+\mu$.}\label{algorithm:dQMAcc}
\begin{algorithmic}[1]
\State The prover sends a state $\rho$ on the quantum register $R_{0,0}$, whose size is $O(\gamma \log(n+r))$, to the left-end node $v_0$ as a proof.
\State The prover sends the quantum registers $R_{j,0}, R_{j,1}$, which are $O((\gamma+\mu) \log(n+r))$ qubits respectively, as proofs to each of the intermediate nodes $v_j$ for $j \in \{1,\ldots,r-1\}$. 
\State The left-end node $v_0$ applies the unitary operation $U_x$ to $\rho$ and $O((\gamma+\mu) \log(n+r))$ ancilla qubits and sends a state $U_x (\rho \otimes \ket{0\cdots 0} \bra{0\cdots 0} ) U_x^\dag $ in the register $R_{0,1}$ to $v_1$.
\State Each intermediate node $v_j$ swaps the states between $R_{j,0}$ and $R_{j,1}$ with probability $\frac{1}{2}$, i.e., symmetrizes the states on $R_{j,0}$ and $R_{j,1}$.
\State Each of the nodes sends its quantum register $R_{j,1}$ to the right neighbor $v_{j+1}$.
\State $v_j$ receives a quantum register from its left neighbor $v_{j-1}$. The node then performs the SWAP test on the registers $(R_{{j-1},1}, R_{j,0})$ and accepts or rejects accordingly.
\State The right-end node $v_r$ receives a state on a register $R_{r-1,1}$ from its left neighbor. Then, $v_r$ performs the POVM measurement $\{M'_{y,0},M'_{y,1}\}$ on the state of $R_{r-1,1}$ and accepts or rejects accordingly. 
\end{algorithmic}
\end{algorithm}

Let us next analyze the completeness and soundness of the protocol $\mathcal{P}_{\QMAcc}$ described in Algorithm \ref{algorithm:dQMAcc}. Note that the analysis is quite close to the analysis of the $\dQMA$ protocol for $\EQ$ on paths.

\subsubsection*{Completeness}
Let us assume that an input $(x,y)$ satisfies $f(x,y)=1$. Then, there exists a quantum proof $\ket{\xi}$ such that $\mathrm{tr} ( M'_{y,1} (U_x (\ket{\xi}\bra{\xi} \otimes \ket{0\cdots 0} \bra{0\cdots 0} ) U_x^\dag ) ) ) = 1$ with probability at least $1-\frac{1}{42n^c r^2}$. 
The prover sends $\sigma = \ket{\xi}\bra{\xi}$ to $v_0$ and $U_x (\ket{\xi}\bra{\xi} \otimes \ket{0\cdots 0} \bra{0\cdots 0} ) U_x^\dag \otimes U_x (\ket{\xi}\bra{\xi} \otimes \ket{0\cdots 0} \bra{0\cdots 0} ) U_x^\dag $ to all the intermediate nodes. Then, all the SWAP tests accept with certainty. Furthermore, the right-end node $v_r$ accepts with probability at least $1-\frac{1}{42n^c r^2}$. Then, from the definition of the completeness of $\dQMA$ protocols, the protocol $\mathcal{P}_{\QMAcc}$ has completeness at least $1-\frac{1}{42n^cr^2}$.

\subsubsection*{Soundness}
\sloppy Let us assume that an input $(x,y)$ satisfies $f(x,y)=0$ and, for any quantum proof $\ket{\xi}$, $\mathrm{tr} ( M'_{y,0} (U_x (\ket{\xi}\bra{\xi} \otimes \ket{0\cdots 0} \bra{0\cdots 0} ) U_x^\dag ) ) ) \geq 1-\frac{1}{42n^cr^2} \geq \frac{2}{3}$. Then, a lemma similar to Lemma \ref{lemma:sum_rej} is shown.

\begin{lemma}
     For $j \in \{1,\ldots,r\}$, let $E_j$ be the event that the local test $v_j$ performs (the SWAP test or the POVM measurement) accepts. Then, $\sum_{j=1}^r \mathrm{Pr}[\neg{E_j}] \geq \frac{4}{81r}$.
\end{lemma}
\begin{proof}
    For conciseness, let us denote $p_j = \mathrm{Pr}[\neg{E_j}]$. By the same discussion to Lemma \ref{lemma:sum_rej}, we have
    \begin{equation}\nonumber
        D( \rho_{0,1}, \rho_{r-1,1} ) \leq 3 \sum_{j=1}^{r-1} \sqrt{p_j},
    \end{equation}
    where $\rho_{0,1}$ is a reduced state on the register $R_{0,1}$ and $\rho_{r-1,1}$ is a reduced state on the register $R_{r-1,1}$.
    From the assumption of the soundness, $\mathrm{tr}(M'_{y,0} (\rho_{0,1})) \geq \frac{2}{3}$. Then, by the same discussion to Lemma \ref{lemma:sum_rej}, we conclude
    \[
        \sum_{j=1}^r p_j \geq \frac{4}{81r}.
    \]
\end{proof}
Using Lemma \ref{lemma:prob}, the protocol $\mathcal{P}_{\QMAcc}$ has soundness $\frac{4}{81r^2}$. Let us again consider a parallel repetition with $O(r^2)$ times as the protocol $P_\pi[k]$ in Algorithm \ref{algorithm:parallel_repetition}, which completes the proof of Theorem \ref{theorem:transform_QMAcc}.
\end{proof}

We next show any $\dQMA$ protocol can be simulated by a $\dQMAsep$ protocol with some overhead. To do it, let us restate the definition of the Linear Space Distance (LSD) problem from \cite{RS04} as a complete problem for $\QMA$ communication protocols. For a subspace $V \subset \mathbb{R}^m$, let us define $S(V) = \{v \in V | \|v\|=1 \}$, the unit sphere in $V$ where $\| \cdot \|$ is the Euclidean norm. For two subspaces $V_1,V_2 \subset \mathbb{R}^m$, let us define 
\[
    \Delta(V_1,V_2) = \min_{v_1 \in S(V_1)} \min_{v_2 \in S(V_2)} \| v_1 - v_2 \|,
\]
as the distance between $V_1$ and $V_2$.
\begin{definition}[The Linear Space Distance (LSD) problem~\cite{RS04}]
    Given two subspaces $V_1$ and $V_2$ of $\mathbb{R}^m$ under the promise that $\Delta(V_1,V_2) \leq 0.1 \cdot \sqrt{2}$ or $\Delta(V_1,V_2) \geq 0.9 \cdot \sqrt{2}$, decide if the distance is small or not.
\end{definition}

\begin{lemma}[Theorem 7 in \cite{RS04}]\label{lemma:reduction_from_QMAcc_to_LSD}
Suppose $f: \mathcal{X} \times \mathcal{Y} \to \{0,1\}$ has a $\nu$-round $\QMA$ communication protocol with a $\gamma$-qubit proof and $\mu$-qubit communications. Then, there exists a mapping from $\mathcal{X}$ and $\mathcal{Y}$ to subspaces of $\mathbb{R} ^ {2 
 ^{(\gamma+\mu)\poly(\nu)}}$\footnote{The dimension of the vector space is different from Theorem 7 in \cite{RS04}, while it is observed by the analysis of the proof. A similar analysis is also considered in \cite{KP14}.}, $x \mapsto A_x$, $y \mapsto B_y$, such that if $f(x,y)=1$, $\Delta(A_x,B_y) \leq 0.1 \cdot \sqrt{2}$ and if $f(x,y)=0$, $\Delta(A_x,B_y) \geq 0.9 \cdot \sqrt{2}$.
\end{lemma}

\begin{lemma}[Theorem 16 in \cite{RS04}]\label{lemma:LSD_QMAcc-complete}
There exists a $\QMA$ one-way communication protocol of cost $O(\log m)$ to solve the LSD problem \footnote{Soundness and completeness do not change in the complex setting of quantum proofs \cite{Mck13}.}.
\end{lemma}

In the definition of the LSD problem, the input precision is infinite. We can define $\mathrm{\widetilde{LSD}}$ as the finite precision version where $\mathbb{R}^m$ are approximated with $O(m^2)$ variables and each variable is described with $O(\log m)$ bits. The input size of the problem is $O(m^2 \log m)$ and the above two results hold for the finite precision analog \cite{RS04}. Therefore we assume that the input size for the LSD problem is  $O(m^2 \log m)$ without loss of generality.

Using the property of the LSD problem as a $\QMA$ communication complete problem, we prove that any $\dQMA$ protocol can be simulated by a $\dQMAsep$ protocol with some overheads.

\begin{theorem}\label{theorem:dQMAsep_for_dQMA}
    Suppose $f: \{0,1\}^n \times \{0,1\}^n \to \{0,1\}$ has a constant-round $\dQMA$ protocol on a path $v_0,\ldots,v_r$, completeness $\frac{2}{3}$, and soundness $\frac{1}{3}$. 
    Let $C:=\sum_{j \in [0,r]} c(v_j) + \min_{j \in [0,r-1]} m(v_j,v_{j+1})$. 
    Then, there exists a $1$-round $\dQMAsep$ protocol for $f$ on the path of length $r$ with completeness $1-\frac{1}{\poly(C 2^{C})}$, soundness $\frac{1}{3}$, using local proof and message of size $\Tilde{O}(r^2 {C}^2)$.
\end{theorem}

\begin{proof}
    Let us denote $j = \underset{i} {\operatorname{argmin}} \, m(v_i,v_{i+1})$. Let us divide $r+1$ nodes into the two groups $v_0,\ldots,v_{j}$ and $v_{j+1},\ldots,v_{r}$. From the original $\dQMA$ protocol, let us consider Alice simulates the protocols of $v_0,\ldots,v_{j}$ and she accepts iff all the parties accept, and Bob simulates the protocols of $v_{j+1},\ldots,v_r$ and he accepts iff all the parties accept. This protocol is a $\QMA^*$ communication protocol whose complexity is at most $C$ to solve $f$. 
    From the inequality~(\ref{equation:relation}) in Section~\ref{subsubsection:QMA_communication}, $\QMAcc(f)$ is at most $2C$. 
    By Lemma \ref{lemma:reduction_from_QMAcc_to_LSD} and Lemma \ref{lemma:LSD_QMAcc-complete}, there exists a $\QMA$ one-way communication protocol of complexity $O(C)$ to solve the LSD problem to which $f$ reduces. 
    Note that the dimension $m$ of the subspaces of the LSD instance is $m=2^{O(C)}$. 
    Since the input size of the LSD is 
    $O(m^2 \log m)=O(C 2^{O(C)})$, Theorem \ref{theorem:transform_QMAcc} implies that 
    there exists a $\dQMAsep$ protocol for the LSD problem (and hence for $f$) on the path of length $r$ with 1-round communication, and completeness $1-\frac{1}{\Omega(\poly(C 2^{C}))}$, soundness $\frac{1}{3}$, using local proof and message sizes 
    $O(r^2 (C) \log (2^{O(C\log (C))}+r)) 
    = \Tilde{O}(r^2 {C}^2)$.
\end{proof}

We also show that there exists an efficient $\dQMA$ protocol for a function which has an efficient $\QMA^*$ communication protocol with some overhead costs.
\begin{proposition}\label{proposition:dQMA_for_QMAcc}
    Suppose $f: \{0,1\}^n \times \{0,1\}^n \to \{0,1\}$ has a $\QMA^*$ communication protocol with cost $C$, i.e., $\QMAcc^*(f)=C$. Then, there exists a $\dQMAsep$ protocol for $f$ on the path of length $r$ with completeness $1-\frac{1}{\poly(n)}$, soundness $\frac{1}{3}$, using local proof and message of size $O(r^2\log (r) \poly(C) )$.
\end{proposition}
\begin{proof}
    From the inequality~(\ref{equation:relation}) in Section~\ref{subsubsection:QMA_communication}, $\QMAcc(f)$ is at most $2C$, and from Lemma \ref{lemma:reduction_from_QMAcc_to_LSD} and Lemma \ref{lemma:LSD_QMAcc-complete}, there exist a $\QMA$ one-way communication protocol of complexity $O( \poly(C))$ to solve the LSD problem to which $f$ reduces, where the dimension $m$ of the subspaces of the LSD instance is $m=2^{\poly(C)}$. Since the input size of the LSD is $O(m^2 \log m)$, Theorem~\ref{theorem:transform_QMAcc} implies the claim described.
\end{proof}
\section{Lower bounds for $\dQMA$ protocols}\label{section:lower_bound}

In this section, we will obtain lower bounds for the size of proofs and communication of $\dQMA$ protocols. In this section, we also focus on the case where the verifier $v_0,\ldots,v_r$ are arranged in a row and the two extremities $v_0$ and $v_r$ have inputs. Let $x \in \{0,1\}^n$ be the input owned by $v_0$, and $y \in \{0,1\}^n$ be the input owned by $v_r$.

\subsection{By a counting argument over quantum states for fooling inputs}\label{subsection:lower_bound_counting}
In this subsection, we will obtain a lower bound of the proof size by a counting argument of quantum states for fooling inputs.

A lower bound for the size of quantum fingerprints of $n$-bits was shown by a reduction to the lower bound of quantum one-way communication complexity for $\EQ$ \cite{BdW01}.
\begin{lemma}[Theorem 8.3.2 in \cite{dW01}]\label{lemma:fingerprint_size}
    Let $\delta \geq 2^{-n}$. Suppose that a family of pure states $\{\ket{h_x}\}_{x \in \{0,1\}^n}$ of $b$-qubit satisfies $| \braket{h_i|h_j} | \leq \delta$ for any distinct $i,j$.  Then, $b = \Omega(\log (\frac{n}{\delta^2} ) )$.
\end{lemma}

\begin{claim}\label{claim:fingerprint_size}
    For any family of sets $S_n$ where $|S_n| \geq s(n)$ and any constant $0 \leq \delta<1$, there exist a sufficiently small constant $c>0$ and large integer $n$ such that, for any family of $c \log  \log s(n) $-qubit pure states $\{\ket{h_x}\}_{x \in S_n}$, there exist $i$ and $j$ such that $| \braket{h_i|h_j} | > \delta$.
\end{claim}
\begin{proof}
    Let us choose a family of sets $S'_n$ so that for all $n$, each set is an arbitrary subset of the set $S_n$ and $|S_n'| = 2^{\lfloor \log s(n) \rfloor}$, and let us correspond an element of $S_n'$ with an element of $\{0,1\}^{\lfloor \log s(n) \rfloor}$ one by one. Then, from Lemma \ref{lemma:fingerprint_size}, if there exists a family of pure states $\{\ket{h_x}\}_{x \in S_n'}$ of $b$-qubit satisfies $| \braket{h_i|h_j} | \leq \delta$ for any distinct elements $i,j \in S_n'$, $b = \Omega(\log \log s(n))$, i.e., for sufficient large $n$ and a constant $c'$, $b \geq c'\log \log s(n)$. Then, for a constant $c < c'$, the claim holds.
\end{proof}

By a counting argument for fooling inputs, we have a lower bound of the proof size of $\dQMAsepsep$ protocols. 

\begin{proposition}\label{proposition:lower_bound_main}
     Let $p \geq 0, \delta>0$, $\nu \in \mathbb{N}$ be constants and $f$ be a Boolean function with a $1$-fooling set of size at least $k$. Let $\mathcal{P}$ be a $\dQMAsepsep$ protocol for $f$ on the path of length $r$ with $\nu$-round communication, completeness $1-p$ and soundness error less than $1-2p-\delta$. 
     Then, for any $i \in [\nu,r-\nu-1]$ and a sufficiently small constant $c$, $\sum_{j=i-\nu+1}^{i+\nu} c(v_j) > c\log \log k $.
\end{proposition}

\begin{proof}
    For conciseness, we prove the case that $\mathcal{P}$ is a 1-round communication protocol (we can easily modify the following proof to the $\nu$-round case).
    
     Let us denote by $S_n$ the 1-fooling set for $f:\{0,1\}^n \times \{0,1\}^n \to \{0,1\}$ where $|S_n| \geq k$. For $x_1,x_2 \in \{0,1\}^n$, let $\ket{\psi_x}$ be a proof with the input $x=(x_1,x_2)$ for all the nodes $v_0,\ldots,v_r$, where $x_1$ is owned by $v_0$ and $x_2$ is owned by $v_r$, and let $\ket{\psi_x}_j$ be a part of the proof for the node $v_j$ where $j=0,\ldots,r$. 
     
    To reach a contradiction, let us assume that $\sum_{j=i}^{i+1} c(v_j) \leq c\log \log k$ for some $i\in [1,r-2]$. Let us consider a family of states $\{\ket{\psi_x}_i \otimes \ket{\psi_x}_{i+1}\}_{(x_1,x_2) \in S_n}$.
     From Claim \ref{claim:fingerprint_size}, since the qubit size of the family is less than or equal to $c\log \log k$ where $c$ is chosen a sufficiently small constant, there exist $y=(y_1,y_2)$ and $z=(z_1,z_2)$ in $S_n$ such that $f(y_1,y_2)=f(z_1,z_2)=1$ and $| \bra{\psi_{y}}_i \otimes \bra{\psi_y}_{i+1} \ket{\psi_z}_i \otimes \ket{\psi_z}_{i+1} |> 1 - \frac{\delta^2}{8}$. From Fact \ref{fact:Fuchs-van} and Fact \ref{fact:tracedistance_contractive} and since the partial trace is a quantum operation, we have
     \begin{eqnarray*}
        D(\ket{\psi_y}_{i},\ket{\psi_z}_i) &\leq& D(\ket{\psi_y}_i \otimes \ket{\psi_y}_{i+1}, \ket{\psi_z}_i \otimes \ket{\psi_z}_{i+1}) \\ &\leq& \sqrt{1-F(\ket{\psi_y}_i \otimes \ket{\psi_y}_{i+1}, \ket{\psi_z}_i \otimes \ket{\psi_z}_{i+1})^2} \\ &=& \sqrt{1-| \bra{\psi_y}_i \otimes \bra{\psi_y}_{i+1} \ket{\psi_z}_i \otimes \ket{\psi_z}_{i+1} |^2} \\
        &<& \sqrt{1-\left(1-\frac{\delta^2}{8}\right)^2} \\
        &=& \sqrt{\frac{\delta^2}{4} - \frac{\delta^4}{16}} \\
        &<& \frac{\delta}{2}.
     \end{eqnarray*}
     We also have $D(\ket{\psi_y}_{i+1},\ket{\psi_z}_{i+1}) < \frac{\delta}{2}$ by the same discussion.
     
     Let $L$ be a register of a part of the proof for $v_0,\ldots,v_i$ and $R$ be a register of the other part (namely, for $v_{i+1},\ldots,v_r$). Let us denote by $\mathsf{out}_j (s,t,\ket{\phi})$ the output of $v_j$ when the input is $(s,t)$ (where $s$ is owned by $v_0$ and $t$ is owned by $v_r$) and the proof is $\ket{\phi}$. From the assumption of the completeness, we have
    \[
        \Pr \big [\bigwedge_{j:j \leq i}\mathsf{out}_j(y_1,y_2,\ket{\psi_y}_{LR})=1 \wedge \bigwedge_{j:j \geq i+1} \mathsf{out}_j(y_1,y_2,\ket{\psi_y}_{LR})=1 \big] \geq 1-p,
    \]
    \[
        \Pr \big [\bigwedge_{j:j \leq i}\mathsf{out}_j(z_1,z_2,\ket{\psi_z}_{LR})=1 \wedge \bigwedge_{j:j \geq i+1} \mathsf{out}_j(z_1,z_2,\ket{\psi_z}_{LR})=1 \big] \geq 1-p.
    \]
    We thus have
    \[
        \Pr \big [\bigwedge_{j:j \leq i}\mathsf{out}_j(y_1,y_2,\ket{\psi_y}_{LR})=1 \big] \geq 1-p,
    \]
    \[
        \Pr \big [\bigwedge_{j:j \geq i+1}\mathsf{out}_j(z_1,z_2,\ket{\psi_z}_{LR})=1 \big] \geq 1-p.
    \]
    
    Let us consider the input assignment $(y_1,z_2)$ combined with the proof assignment $\ket{\psi_y}_L \otimes \ket{\psi_z}_R$. By the definition of the $1$-fooling set, $f(y_1,z_2)=0$ without loss of generality. Since the protocol $\mathcal{P}$ has 1-round in the verification algorithm and the proofs are separable, we have
    \[
        \Pr \big [\bigwedge_{j:j \leq {i-1}}\mathsf{out}_j(y_1,y_2,\ket{\psi_y}_{LR})=1 \big] = \Pr \big [\bigwedge_{j:j \leq {i-1}}\mathsf{out}_j(y_1,z_2,\ket{\psi_y}_L \otimes \ket{\psi_z}_R)=1 \big],
    \]
    \[
        \Pr \big [\bigwedge_{j:j \geq i+2}\mathsf{out}_j(z_1,z_2,\ket{\psi_z}_{LR})=1 \big] = \Pr \big [\bigwedge_{j:j \geq i+2}\mathsf{out}_j(y_1,z_2,\ket{\psi_y}_L \otimes \ket{\psi_z}_R)=1 \big].
    \]
    The output of the node $v_i$ can be only affected by $\ket{\psi_y}_{i-1} \otimes \ket{\psi_y}_i \otimes \ket{\psi_z}_{i+1}$ and the binary string $y_1$. Similarly the output of the node $v_{i+1}$ can be affected by $\ket{\psi_y}_{i} \otimes \ket{\psi_z}_{i+1} \otimes \ket{\psi_z}_{i+2}$ and the binary string $z_2$. With Fact \ref{fact:distinguishability_quantumalgo}, we thus have
    \begin{eqnarray*}
        \lefteqn{| \Pr [\mathsf{out}_i(y_1,y_2,\ket{\psi_y}_{LR})=1 \big] - \Pr [\mathsf{out}_i(y_1,z_2,\ket{\psi_y}_L \otimes \ket{\psi_z}_R)=1 ] |} \\ && \leq D(\ket{\psi_y}_{i-1} \otimes \ket{\psi_y}_i \otimes \ket{\psi_y}_{i+1},\ket{\psi_y}_{i-1} \otimes \ket{\psi_y}_i \otimes \ket{\psi_z}_{i+1}) = D(\ket{\psi_y}_{i+1},\ket{\psi_z}_{i+1}) < \frac{\delta}{2}, \\
        \lefteqn{| \Pr [\mathsf{out}_{i+1}(z_1,z_2,\ket{\psi_z}_{LR})=1 \big] - \Pr [\mathsf{out}_{i+1}(y_1,z_2,\ket{\psi_y}_L \otimes \ket{\psi_z}_R)=1 ] |} \\ && \leq D(\ket{\psi_y}_{i} \otimes \ket{\psi_z}_{i+1} \otimes \ket{\psi_z}_{i+2} ,\ket{\psi_z}_i \otimes \ket{\psi_z}_{i+1} \otimes \ket{\psi_z}_{i+2}) = D(\ket{\psi_y}_{i} ,\ket{\psi_z}_i) < \frac{\delta}{2}.
    \end{eqnarray*}
    Combining the inequalities and the union bound, we have
    \begin{eqnarray*}
        \lefteqn{ \Pr \big [\bigwedge_{j:j \leq i}\mathsf{out}_j(y_1,z_2,\ket{\psi_y}_L \otimes \ket{\psi_z}_R)=1 \wedge \bigwedge_{j:j \geq i+1} \mathsf{out}_j(y_1,z_2,\ket{\psi_y}_L \otimes \ket{\psi_z}_R)=1 \big] 
        } \\
        &=& \Pr \Biggl[ \left(\bigwedge_{j:j \leq {i-1}}\mathsf{out}_j(y_1,y_2,\ket{\psi_y})=1 \wedge \mathsf{out}_i(y_1,z_2,\ket{\psi_y}_L \otimes \ket{\psi_z}_R)=1 \right) \wedge \\ && \hspace{60pt} \left(\mathsf{out}_{i+1}(y_1,z_2,\ket{\psi_y}_L \otimes \ket{\psi_z}_R)=1 \wedge \bigwedge_{j:j \geq i+2} \mathsf{out}_j(z_1,z_2,\ket{\psi_z})=1 \right) \Biggr] \\
        &\geq& 1 - \Pr \left [\lnot \left( \bigwedge_{j:j \leq {i-1}}\mathsf{out}_j(y_1,y_2,\ket{\psi_y})=1 \wedge \mathsf{out}_i(y_1,z_2,\ket{\psi_y}_L \otimes \ket{\psi_z}_R)=1 \right) \right]\\ && \hspace{20pt} - \Pr \left [\lnot \left(\mathsf{out}_{i+1}(y_1,z_2,\ket{\psi_y}_L \otimes \ket{\psi_z}_R)=1 \wedge \bigwedge_{j:j \geq i+2} \mathsf{out}_j(z_1,z_2,\ket{\psi_z})=1 \right) \right] \\
        &\geq& 1-\delta-\Pr \left [\lnot \left( \bigwedge_{j:j \leq {i-1}}\mathsf{out}_j(y_1,y_2,\ket{\psi_y})=1 \wedge \mathsf{out}_i(y_1,y_2,\ket{\psi_y})=1 \right) \right]\\ && \hspace{40pt} - \Pr \left [\lnot \left(\mathsf{out}_{i+1}(z_1,z_2,\ket{\psi_z})=1 \wedge \bigwedge_{j:j \geq i+2} \mathsf{out}_j(z_1,z_2,\ket{\psi_z})=1 \right) \right] \\
        &\geq& 1-2p-\delta,
    \end{eqnarray*}
    which contradicts the condition of the soundness. Therefore, we conclude $\sum_{j=i}^{i+1} c(v_j) > c\log \log k$ for any $i\in[1,r-2]$.
\end{proof}

The proposition above implies that any $\dQMAsepsep$ protocol for $\EQ$ and $\GT$ with sufficiently high completeness and low soundness error requires $\Omega(r\log n)$-qubit quantum proofs.

\begin{theorem}\label{theorem:quantum_lower_bound_separable}
    Let $p \geq 0, \delta>0, \nu \in \mathbb{N}$ be constants and $f:(\{0,1\}^n)^2\rightarrow\{0,1\}$ be a Boolean function with a $1$-fooling set of size $2^n$ (including $\EQ$ and $\GT$). Let $\mathcal{P}$ be a $\dQMAsepsep$ protocol for $f$ on the path of length $r$ with $\nu$-round communication, completeness $1-p$ and soundness error less than $1-2p-\delta$. 
    Then, $\sum_{j=0}^{r} c(v_j) = \Omega(r\log n)$.
\end{theorem}
\begin{proof}
    Assume that $\sum_{j=0}^{r} c(v_j) \leq \lfloor \frac{r-1}{2\nu} \rfloor \lfloor c \log n \rfloor$ for a sufficiently small constant $c$. Then, by the pigeonhole principle, there exists $i \in [\nu,r-\nu-1]$ such that $\sum_{j=i-\nu+1}^{i+\nu} c(v_j) \leq c\log n$, which contradicts Proposition \ref{proposition:lower_bound_main}. Therefore, $\sum_{j=0}^{r} c(v_j) = \Omega(r\log n)$.
\end{proof}

Even for entangled proofs, we obtain the following lower bound by combining Theorem~\ref{theorem:quantum_lower_bound_separable} with Theorem~\ref{theorem:dQMAsep_for_dQMA}. 

\begin{theorem}\label{theorem:quantum_lower_bound_entangled1}
    Let $f:(\{0,1\}^n)^2\rightarrow\{0,1\}$ be a Boolean function with a $1$-fooling set of size $2^n$ (including $\EQ$ and $\GT$). Let $\mathcal{P}$ be a $\dQMA$ protocol for $f$ on the path of length $r$ with constant-round communication, completeness $\frac{2}{3}$ and soundness $\frac{1}{3}$. 
    Let $C := \sum_{j} c(v_j) + \min_{j \in [0,r-1]} m(v_j,v_{j+1})$. Then, $\mathcal{P}$ satisfies $C = \Omega(\frac{(\log n)^{1/2-\epsilon}}{r^{1+\epsilon'}})$ for any constants $\epsilon,\epsilon'>0$.
\end{theorem}
\begin{proof}
    Assume that $\mathcal{P}$ satisfies
    $C = o(\frac{(\log n)^{1/2-\epsilon}}{r^{1+\epsilon'}})$. 
    Then, from Theorem~\ref{theorem:dQMAsep_for_dQMA}, there exists a $\dQMAsep$ (and hence $\dQMAsepsep$) protocol for $f$ on the path of length $r$ with 1-round communication, completeness $\frac{3}{4}$ and soundness $\frac{1}{3}$ and total proof size $$\sum_{j=0}^{r} c(v_j) = \Tilde{O}\bigg(r^3
    \bigg(
    \frac{(\log n)^{1/2-\epsilon}}{r^{1+\epsilon'}}\bigg)^2\bigg)
    =\Tilde{O}(r^{1-2\epsilon'}(\log n)^{1-2\epsilon})
    =o(r\log n),
    $$ 
    which contradicts Theorem \ref{theorem:quantum_lower_bound_separable}.
\end{proof}

For entangled proofs, we can have another lower bound. 

\begin{lemma}\label{lemma:entangled_lowerbound_main}
     Let $\nu \in \mathbb{N}$ be a constant and $f$ be a function which has a $1$-fooling set of size at least $2$. Let $\mathcal{P}$ be a $\nu$-round $\dQMA$ protocol for $f$ on the path of length $r$ with a proof of size satisfying $\sum_{j=i-\nu+1}^{i+\nu} c(v_j) = 0$ for $i \in [\nu,r-\nu-1]$, and completeness $1-p$. Then, $\mathcal{P}$ has soundness error at least $1-2p$.
\end{lemma}

\begin{proof}
For conciseness, we prove the case that $\mathcal{P}$ is a 1-round communication protocol (we can easily modify the following proof to the $\nu$-round case).

     Let $(x,y)$ and $(x',y')$ be in the 1-fooling set for $f$, i.e., $f(x,y)=1,f(x',y')=1$ and $f(x,y')=0$ without loss of generality. Let $\ket{\psi}$ be a proof with the input $(x,y)$ for all the nodes $v_0,\ldots,v_r$ and let $\ket{\psi'}$ be a proof with the input $(x',y')$ for all the nodes $v_0,\ldots,v_r$. 
     
     From the assumption of the completeness, we have
     \[
        \Pr \big [\bigwedge_{j:j \leq i}\mathsf{out}_j(x,y,\ket{\psi}_{LR})=1 \wedge \bigwedge_{j:j \geq i+1} \mathsf{out}_j(x,y,\ket{\psi}_{LR})=1 \big] \geq 1-p,
     \]
     \[
        \Pr \big [\bigwedge_{j:j \leq i}\mathsf{out}_j(x',y',\ket{\psi'}_{LR})=1 \wedge \bigwedge_{j:j \geq i+1} \mathsf{out}_j(x',y',\ket{\psi'}_{LR})=1 \big] \geq 1-p.
     \]
     
     From Fact \ref{fact:Schmidt}, $\ket{\psi}_{LR} = \sum_j \sqrt{p_{j}}\ket{\psi_{j}}_L \ket{\phi_{j}}_R$ and $\ket{\psi'}_{LR} = \sum_j \sqrt{p_{j}'}\ket{\psi'_{j}}_L \ket{\phi'_{j}}_R$. Let $\rho = \mathrm{tr}_R \ket{\psi}\bra{\psi}_{LR} = \sum_j p_{j} \ket{\psi_{j}}\bra{\psi_{j}}$, $\sigma = \mathrm{tr}_L \ket{\psi}\bra{\psi}_{LR} = \sum_j p_{j} \ket{\phi_{j}}\bra{\phi_{j}}$, $\rho' = \mathrm{tr}_R \ket{\psi'}\bra{\psi'}_{LR} = \sum_j p_{j}' \ket{\psi'_{j}} \bra{\psi'_{j}}$ and $\sigma' = \mathrm{tr}_L \ket{\psi'}\bra{\psi'}_{LR} = \sum_j p_{j}' \ket{\phi'_{j}} \bra{\phi'_{j}}$. Let us consider the case where the input are distinct $x$ and $y'$ and the proof $\rho \otimes \sigma'$. Since $\sum_{j=i}^{i+1} c(v_j) = 0$ and the protocol $\mathcal{P}$ has only 1-round communication,
     \[
        \Pr \big [\bigwedge_{j:j \leq i}\mathsf{out}_j(x,y,\ket{\psi}_{LR})=1 \big]=
        \Pr \big [\bigwedge_{j:j \leq i}\mathsf{out}_j(x,y,\rho \otimes \sigma)=1  \big] = \Pr \big [\bigwedge_{j:j \leq i}\mathsf{out}_j(x,y',\rho \otimes \sigma')=1  \big] 
     \]
     \[
        \Pr \big [ \bigwedge_{j:j \geq i+1} \mathsf{out}_j(x',y',\ket{\psi'}_{LR})=1 \big]=
        \Pr \big [\bigwedge_{j:j \geq i+1}\mathsf{out}_j(x',y',\rho' \otimes \sigma')=1  \big] = \Pr \big [\bigwedge_{j:j \geq i+1}\mathsf{out}_j(x,y',\rho \otimes \sigma')=1  \big] 
     \]
     
     Therefore, we have
     \begin{eqnarray*}
        \lefteqn{ \Pr \big [\bigwedge_{j:j \leq i}\mathsf{out}_j(x,y',\rho \otimes \sigma')=1 \wedge \bigwedge_{j:j \geq i+1} \mathsf{out}_j(x,y',\rho \otimes \sigma')=1 \big]} \\
        &\geq&
        1 - \Pr \big [\lnot\bigwedge_{j:j \leq i}\mathsf{out}_j(x,y',\rho \otimes \sigma')=1 \big]
        - \Pr \big [\lnot\bigwedge_{j:j \geq i+1} \mathsf{out}_j(x,y',\rho \otimes \sigma')=1 \big] \\
        &=& 1 - \Pr \big [\lnot \bigwedge_{j:j \leq i}\mathsf{out}_j(x,y,\ket{\psi})=1 \big] - \Pr \big [\lnot \bigwedge_{j:j \geq i+1}\mathsf{out}_j(x',y',\ket{\psi'})=1 \big] \\
        &\geq& 1-2p,
    \end{eqnarray*}
    as claimed.
\end{proof}

\begin{proposition}\label{proposition:quantum_lower_bound_entangled}
     Let $f$ be a function which has a $1$-fooling set of size at least $2$. Let $\mathcal{P}$ be a $\nu$-round $\dQMA$ protocol for $f$ on the path of length $r$ with a proof of size satisfying $\sum_{j=0}^{r} c(v_j) \leq \lfloor \frac{r-1}{2\nu} \rfloor -1$, and completeness $1-p$. Then, $\mathcal{P}$ has soundness error at least $1-2p$.
\end{proposition}

\begin{proof}
    From the pigeonhole principle, if $\sum_{j=i}^{i+1} c_j \leq \lfloor \frac{r-1}{2\nu} \rfloor -1$, there exists $i \in [1,r-2]$ such that $\sum_{j=i}^{i+1} c(v_j) = 0$. Then, from Lemma \ref{lemma:entangled_lowerbound_main}, we have the claim.
\end{proof}

\begin{corollary}\label{corollary:quantum_lower_bound_entangled}
    Let $f^+:(\{0,1\}^n)^2\rightarrow\{0,1\}$ be any non-constant Boolean function.  Let $\mathcal{P}$ be a constant-round $\dQMA$ protocol for $f^+$ on the path of length $r$ with completeness $1-p$ and soundness error at least $1-2p$. Then $\mathcal{P}$ satisfies $\sum_{j=0}^{r} c(v_j) = \Omega(r)$.
\end{corollary}

Combining the two lower bounds on entangled proofs, we have a lower bound below.

\begin{theorem}\label{theorem:quantum_lower_bound_entangled2}
    Let $f:(\{0,1\}^n)^2\rightarrow\{0,1\}$ be a Boolean function with a $1$-fooling set of size $2^n$ (including $\EQ$ and $\GT$). Let $\mathcal{P}$ be a constant-round $\dQMA$ protocol for $f$ on the path of length $r$ with completeness $\frac{3}{4}$ and soundness $\frac{1}{4}$. 
    Then, $\mathcal{P}$ satisfies $\sum_{j=0}^{r} c(v_j) + \min_{j\in [0,r-1]} m(v_j,v_{j+1})= \Omega((\log n)^{1/4-\epsilon})$ for any constant $\epsilon>0$.
\end{theorem}

\begin{proof}
    From Theorem \ref{theorem:quantum_lower_bound_entangled1} and Corollary \ref{corollary:quantum_lower_bound_entangled}, $\sum_{j=0}^{r} c(v_j) + \min_{j\in [0,r-1]} m(v_j,v_{j+1}) \geq \sum_{j=0}^{r} c(v_j) = \Omega(r)$ and $\sum_{j=0}^{r} c(v_j) + \min_{j\in [0,r-1]} m(v_j,v_{j+1})=\Omega(\frac{(\log n)^{1/2-\epsilon}}{r^{1+\epsilon'}})$ for any constants $\epsilon,\epsilon'>0$. Since for any constant $\epsilon''>0$, there exist $\epsilon,\epsilon'>0$ such that $\max\{r, \frac{(\log n)^{1/2-\epsilon}}{r^{1+\epsilon'}}\} \geq (\log n)^{1/4-\epsilon''}$ for any $r$, we have $\sum_{j=0}^{r} c(v_j) + \min_{j\in [0,r-1]} m(v_j,v_{j+1})=\Omega((\log n)^{1/4-\epsilon''})$ for any constant $\epsilon''>0$.
\end{proof}

\subsection{By a reduction to lower bounds of $\QMA$ communication protocols}\label{subsection:lower_bound_reduction}
In this subsection, we prove lower bounds of $\dQMA$ protocols by a reduction to a lower bound of the two nodes case. 

Klauck \cite{Kla11} derived lower bounds on $\QMA$ communication protocols for some predicates. To prove the lower bounds, he first observed the proof efficient error reduction for $\QMA$ \cite{MW05} works for the $\QMA$ communication protocols as well. Then, after the error reduction, he considered to replace a proof with a maximally entangled state (which can be generated by Alice) and have an unbounded error communication protocol. Klauck finally derived a quantum communication lower bound for such the unbounded-error communication protocol exploiting the one-sided smooth discrepancy \cite{Kla11}. 

Let us denote by $\mathsf{sdisc}^1 (f)$ the one-sided smooth discrepancy of a function $f$ and see Definition 8 and 9 in \cite{Kla11} for the definition of the one-sided smooth discrepancy.

\begin{lemma}[Theorem 2 in \cite{Kla11}]
    $\QMAcc(f) = \Omega \bigg( \sqrt{\log \mathsf{sdisc}^1 (f) } \bigg)$.
\end{lemma}

\begin{definition}[Disjointness]
The disjoint function $\mathsf{DISJ}$ receives two $n$-bit strings $x$ and $y$ as inputs. $\mathsf{DISJ}(x,y) := \bigwedge_{i=1,\ldots,n}(\lnot x_i \lor \lnot y_i)$.
\end{definition}

\begin{corollary}[Theorem 1 in \cite{Kla11}]
$\QMAcc(\mathsf{DISJ})=\Omega(n^\frac{1}{3})$.
\end{corollary}

\begin{definition}
    The inner product function receives two $n$-bit strings $x$ and $y$ as inputs. $\mathsf{IP}_2 (x,y) = \bigoplus_{i=1,\ldots,n} (x_i \land y_i)$.
\end{definition}

\begin{lemma}[Corollary 1 in \cite{Kla11}]
    $\QMAcc(\mathsf{IP}_2)=\Omega(n^\frac{1}{2})$.
\end{lemma}

Sherstov \cite{She11} introduced the pattern matrices, which is a method to convert a Boolean function into a hard communication problem.

\begin{definition}[Pattern Matrices, Definition 5 in \cite{Kla11}]
    For a function $f:\{0,1\}^n \to \{0,1\}$, the pattern matrix $P_f$ is the communication matrix of the following problem: Alice receives a bit string $x$ of length $2n$, Bob receives two bit strings $y$, $z$ of length $n$ each. The output of the function described by $P_f$ on inputs $x,y,z$ is $f(x(y) \oplus z)$, where $\oplus$ is the bitwise xor, and $x(y)$ denotes the $n$ bit string that contains $x_{2i-y_i}$ in position $i=1,\ldots,n$.
\end{definition}

$\mathsf{AND}$ function is defined by $\mathsf{AND}(x_1,\ldots,x_n) = x_1 \land \cdots \land x_n$.

\begin{lemma}[Corollary 2 in \cite{Kla11}]
$\QMAcc(P_\mathsf{AND})=\Omega(n^\frac{1}{3})$.
\end{lemma}

We observe that the result and proof strategy of \cite{Kla11} still hold for $\QMA^*$ communication protocols. One reason is that a maximally mixed state over Alice and Bob is a separable state between Alice and Bob and it can be produced by Alice and Bob with no communication. Another reason is the proof-efficient error reduction of the $\QMA$ communication protocols from \cite{MW05} also holds for the $\QMA^*$ communication protocols. Moreover, the rest of the proof is the same for such an unbounded-error communication protocol,  obtaining a quantum lower bound.

\begin{fact}
    Assume that there exists a $\QMA^*$ communication protocol with proof length $\gamma_1$ and $\gamma_2$ and communication length $\mu$ with bounded error. Then, there exists a $\QMA^*$ communication protocol with proof length $\gamma_1$ and $\gamma_2$ and communication length $O(\mu \cdot k)$ and error $\frac{1}{2^k}$.
\end{fact}

\begin{claim}\label{claim:sdisc_lowerbound}
    $\QMAcc^*(f) = \Omega \bigg(\sqrt{\log \mathsf{sdisc}^1 (f) } \bigg)$.
\end{claim}

\begin{corollary}\label{claim:quantum_lowerbound_DISJ}
    $\QMAcc^*(\mathsf{DISJ})=\Omega(n^\frac{1}{3})$,
    $\QMAcc^*(\mathsf{IP}_2)=\Omega(n^\frac{1}{2})$,
    $\QMAcc^*(P_\mathsf{AND})=\Omega(n^\frac{1}{3})$.
\end{corollary}

Then, we obtain a lower bound of $\dQMA$ by a reduction to the lower bound of Claim \ref{claim:sdisc_lowerbound}.

\begin{theorem}\label{theorem:quantum_lower_bound_reduction}
    Assume that $\mathcal{P}$ is a $\dQMA$ protocol on the path of length $r$ with arbitrary rounds to solve $f$ with completeness $\frac{2}{3}$ and soundness $\frac{1}{3}$. Then, $\mathcal{P}$ satisfies $\sum_{j=0}^{r} c(v_j) + \min_{j\in [0,r-1]} m(v_j,v_{j+1})= \Omega(\sqrt{\log \mathsf{sdisc}^1 (f) })$.
\end{theorem}

\begin{proof}
Let us consider reductions from a $\dQMA$ protocol to a $\QMA^*$ communication protocol in (slightly) different ways depending on how we split all the nodes into two groups. Let us name each reduction an $i$-th reduction when we consider that $v_0,\ldots,v_i$ is one set of nodes and $v_{i+1},\ldots,v_r$ is the other set for $i \in \{0,\ldots,r-1\}$ and the reductions can be described as Algorithm \ref{algorithm:reduction}. The $\QMA^*$ communication protocol after the reductions has completeness $\frac{2}{3}$ and soundness $\frac{1}{3}$ and its complexity is $\sum_{j=0}^r c(v_j) + m(v_i,v_{i+1})$. The complexity must be $\Omega \bigg(\sqrt{\log \mathsf{sdisc}^1 (f) } \bigg)$ for all $i$ from Claim \ref{claim:sdisc_lowerbound}, which implies the claim.

\begin{algorithm}[H]
\caption{\, $i$-th reduction from a $\dQMA$ protocol to a $\QMA^*$ communication protocol}
\label{algorithm:reduction}
\begin{algorithmic}[1]
\State Alice receives a $\biggl(\sum_{j=0}^i c(v_j)\biggr)$-qubit state and Bob receives a $\biggl(\sum_{j=i+1}^r c(v_j)\biggr)$-qubit state from a prover (Merlin).
\State Alice simulates the computation and communication of the nodes $v_0,\ldots,v_i$ communicating with Bob by $m(v_i,v_{i+1})$ qubits. Bob simulates the computation and communication of the nodes $v_{i+1},\ldots,v_r$ communicating with Alice by $m(v_i,v_{i+1})$. 
\State Alice accepts if and only if all the nodes $v_0,\ldots,v_i$ accept. Bob accepts if and only if all the nodes $v_{i+1},\ldots,v_r$ accept.
\end{algorithmic}
\end{algorithm}

\end{proof}

For concrete functions, we have lower bounds from Theorem \ref{theorem:quantum_lower_bound_reduction}.

\begin{corollary}
    Assume that $\mathcal{P}$ is a $\dQMA$ protocol on the path of length $r$ with arbitrary rounds to solve $\mathsf{DISJ}$ with completeness $\frac{2}{3}$ and soundness $\frac{1}{3}$. Then, $\mathcal{P}$ satisfies $\sum_{j=0}^{r} c(v_j) + \min_{j\in [0,r-1]} m(v_j,v_{j+1})= \Omega(n^\frac{1}{3})$.
\end{corollary}

\begin{corollary}
    Assume that  $\mathcal{P}$ is a $\dQMA$ protocol on the path of length $r$ with arbitrary rounds to solve $\mathsf{IP}_2$ with completeness $\frac{2}{3}$ and soundness $\frac{1}{3}$. Then, $\mathcal{P}$ satisfies $\sum_{j=0}^{r} c(v_j) + \min_{j\in [0,r-1]} m(v_j,v_{j+1})= \Omega(n^\frac{1}{2})$.
\end{corollary}

\begin{corollary}
    Assume that $\mathcal{P}$ is a $\dQMA$ protocol on the path of length $r$ with arbitrary rounds to solve $P_\mathsf{AND}$ with completeness $\frac{2}{3}$ and soundness $\frac{1}{3}$. Then, $\mathcal{P}$ satisfies $\sum_{j=0}^{r} c(v_j) + \min_{j\in [0,r-1]} m(v_j,v_{j+1})= \Omega(n^\frac{1}{3})$.
\end{corollary}

\section*{Acknowledgements} Part of the work was done while AH was visiting Nagoya University and the Institute for Quantum Computing, University of Waterloo, and AH is grateful to their hospitality. AH would like to thank Richard Cleve, Fran{\c{c}}ois Le Gall, Masayuki Miyamoto, Yuki Takeuchi, Seiichiro Tani and Eyuri Wakakuwa for helpful discussions.

AH is supported by JSPS KAKENHI grants Nos. JP22J22563 and NICT Quantum Camp 2023. SK is funded by the Natural Sciences and Engineering Research Council of Canada (NSERC) Discovery Grants Program and Fujitsu Labs America. Research at the Institute for Quantum Computing (IQC) is supported by Innovation, Science and Economic Development (ISED) Canada. HN is supported by the JSPS KAKENHI grants JP19H04066, JP20H05966, JP21H04879, JP22H00522 and by the MEXT Q-LEAP grants JPMXS0120319794.

\bibliographystyle{alpha}
\bibliography{ref}

\end{document}